\makeatletter
\def\cl@chapter{}
\makeatother
\documentclass[onecolumn]{svjour3}          
\smartqed 

\usepackage{stmaryrd}
\usepackage{mdframed}
\usepackage{mathpartir}
\usepackage{amsmath}
\usepackage{amssymb,makecell}
\usepackage{xcolor}
\usepackage[bookmarks=true, colorlinks=true, linkcolor=blue!50!black,
citecolor=orange!50!black, urlcolor=orange!50!black, pdfencoding=unicode]{hyperref}
\usepackage[color,all,poly,matrix,arrow]{xy}
\usepackage{mathtools}
\usepackage[capitalize]{cleveref}
\usepackage{pgfplots}
\usepackage[ruled,linesnumbered]{algorithm2e}
\usepackage{setspace}
\usepackage{tikz}
\usetikzlibrary{arrows}
\usetikzlibrary{cd}

\DeclareMathOperator*{\colim}{colim}

\newcommand{\LTO}[1]{\overset{\sf{#1}}{\bullet}}
\newcommand{\DTO}[1]{\underset{\sf{#1}}{\bullet}}
\newcommand{\DTOW}[1]{\color{white}{\underset{\sf{#1}}{\bullet}}}

\newcommand{\init}{\mathit{init}}

\newcommand{\ax}{\mathfrak{A}}

\newcommand{\To}[1]{\xrightarrow{#1}}
\newcommand{\Too}[1]{\xhookrightarrow{\ \ #1\ \ }}

\newcommand{\pullbackcorner}[1][dl]{\save*!/#1-1pc/#1:(-1,1)@^{|-}\restore}

\newcommand{\myvdots}{\raisebox{.006\baselineskip}{\ensuremath{\vdots}}}

\newcommand{\Set}{{\sf Set}}
\makeatletter
\renewcommand\tableofcontents{%
  \section*{\contentsname}%
  \begingroup
    \small
    \@starttoc{toc}%
  \endgroup
}
\makeatother

\begin{document}

\title{Fast Left Kan Extensions Using The Chase}
%
%
\author{Joshua Meyers \and David I. Spivak \and Ryan Wisnesky (corresponding author)}

%
%

\institute{All authors at Conexus AI}

\date{\today} 

\maketitle              
\begin{abstract}
We show how computation of left Kan extensions can be reduced to computation of free models of cartesian (finite-limit) theories.  We discuss how the standard and parallel chase compute weakly free models of regular theories and free models of cartesian theories, and compare the concept of ``free model'' with a similar concept from database theory known as ``universal model''.  We prove that, as algorithms for computing finite free models of cartesian theories, the standard and parallel chase are complete under fairness assumptions.  Finally, we describe an optimized implementation of the parallel chase specialized to left Kan extensions that achieves an order of magnitude improvement in our performance benchmarks compared to the next fastest left Kan extension algorithm we are aware of.  
\end{abstract}

\keywords{
Computational category theory \and left Kan extensions \and the Chase \and Data migration \and Data integration \and Lifting Problems \and Regular Logic \and Existential Horn Logic \and Datalog-E \and Model Theory
}

\section{Introduction}

{\it Left Kan extensions}~\cite{CARMODY1995459} are used for many purposes in automated reasoning: to enumerate the elements of finitely-presented algebraic structures such as  monoids; to construct semi-decision procedures for Thue (equational) systems; to compute the cosets of groups; to compute the orbits of a group action; to compute quotients of sets by equivalence relations; and more.  

Left Kan extensions are described category-theoretically, and we assume a knowledge of category theory~\cite{BW} in this paper, but see the next section for a review.  Let $C$ and $D$ be categories and $F: C\to D, I:C\to {\sf Set}$ be functors.  The left Kan extension (formally defined in Example~\ref{leftkan}) $\Sigma_F(I):D\to {\sf Set}$ always exists when $C$ is small\footnote{A benign set-theoretic assumption to avoid set-of-all-sets paradoxes.} and is unique up to unique isomorphism, but it need not be finite ($\Sigma_F(I)(d)$ need not have finite cardinality for any $d$).  In this paper we describe how to compute finite left Kan extensions when $C$, $D$, and $F$ are finitely presented and $I$ is finite, a semi-computable problem originally solved in \cite{CARMODY1995459} and significantly improved upon  in~\cite{BUSH2003107}.




\subsection{Motivation} 

Our interest in left Kan extensions comes from their use in data migration~\cite{wadt,relfound,DBLP:journals/jfp/SchultzW17}, where $C$ and $D$ represent database schemas, $F$ represents a ``schema mapping''~\cite{Haas:2005:CGU:1066157.1066252} defining a translation from schema $C$ to $D$, and $I$ represents an input $C$-database (often called an {\it instance}) that we wish to migrate to schema $D$.  Our implementation of the fastest left Kan algorithm we knew of from existing literature~\cite{BUSH2003107} was impractical for large input instances, yet it bore a striking operational resemblance to an algorithm from relational database theory known as the {\it chase}~\cite{Deutsch:2008:CR:1376916.1376938}, which is also used to solve data migration problems, and for which efficient implementations are known~\cite{Benedikt:2017:BC:3034786.3034796}.  The chase takes an input instance $\mathcal{I}$ and a set of formulae $\mathcal{F}$ in a subset of first-order logic known to logicians as existential Horn logic~\cite{Deutsch:2008:CR:1376916.1376938}, to category theorists as {\it regular logic}~\cite{relolog}, to database theorists as datalog-E and/or {\it embedded dependencies}~\cite{Deutsch:2008:CR:1376916.1376938}, and to topologists as lifting problems~\cite{spivak2014}, and constructs an $\mathcal{F}$-model $chase_\mathcal{F}(\mathcal{I})$ that is ``universal'' among other such ``$\mathcal{F}$-repairs'' of $\mathcal{I}$.

\subsection{Related Work}

In this paper, we show how left Kan extensions can be computed by way of constructing a free model of a cartesian theory on a given instance.  As described in the next paragraph, construction of the free model of a cartesian theory on a given instance resembles the classical {\it universal model construction}~\cite{Deutsch:2008:CR:1376916.1376938} in database theory, except for three important technical differences:
\begin{enumerate}
    \item In database theory, databases are assumed to contain two disjoint kinds of value, constants and labelled nulls, with database homomorphisms required to preserve constants.  In this terminology, the databases that we are left Kan extending are always assumed to be made up entirely of labelled nulls.    
    \item In database theory, the ``universal model'' solution concept is prevalent; whereas in category theory, the ``free model'' and ``weakly free model'' solution concepts are prevalent.  We will compute left Kan extensions by way of free models rather than universal models.
    \item In database theory, theories are typically assumed to be ``regular'', i.e., in a $\forall x \exists y P(x,y)$ form.  The theories we require for computing left Kan extensions are always ``cartesian'', i.e., in $\forall x \exists ! y P(x,y)$ (exists unique) form.
\end{enumerate}

\subsection{Contributions} 

In this paper, we:

\begin{itemize}

\item show how the problem of computing left Kan extensions of set-valued functors can be reduced to the problem of computing free models of {\it cartesian} theories~\cite{adamek_rosicky_1994} (regular theories where every $\exists$ quantifier is read as ``exists-unique'') on input instances; and,


\item prove that the standard chase and parallel chase~\cite{Deutsch:2008:CR:1376916.1376938} compute finite weakly free models of regular theories and finite free models of cartesian theories; and,

\item prove completeness of the standard and parallel chase on cartesian theories under fairness assumptions; and,


\item describe an optimized left Kan extension algorithm, inspired by the parallel chase, that achieves an order of magnitude improvement in our performance benchmarks compared to the next fastest left Kan extension algorithm we are aware of~\cite{BUSH2003107}.

\end{itemize}



 \subsection{Outline}

This paper is structured as follows.  In the next section we review category theory~\cite{BW} and then describe a running example of a left Kan extension.  In section~\ref{lk} we show that left Kan extensions can be considered as free models of cartesian theories.  In section~\ref{chasecart} we discuss how chase algorithms can be used for computing such free models of cartesian theories, as well as the more general case of weakly free models of regular theories.  
In section~\ref{impl} we describe our particular left Kan algorithm implementation, compare it to the algorithm in~\cite{BUSH2003107}, and provide experimental performance results.  We conclude in section~\ref{comp} by discussing additional differences between the chase as used in relational database theory and as used in this paper.  We assume knowledge of formal logic and algebraic specification at the level of~\cite{Baader:1998:TR:280474}, and knowledge of left Kan extensions at the level of~\cite{CARMODY1995459} and knowledge of the chase at the level of~\cite{Deutsch:2008:CR:1376916.1376938} is helpful.

\subsection{Review of Category Theory}
\label{section.ct}

In this section, we review standard definitions and results from category theory~\cite{BW}.  We make the technical distinction between ``{\it class}'' and ``{\it set}'' -- all sets are classes, but not all classes are sets.  This distinction allows us to speak of ``the class of all sets'', whereas invoking ``the set of all sets'' would run into Cantor's paradox.  A {\it class function} is defined similarly to a function, except that it uses the word ``class'' where the definition of ``function'' uses the word ``set''.

\begin{definition} A {\it quiver}, (aka directed multi-graph) $Q$ consists of
 a class ${\sf Ob}(Q)$, the members of which we call {\it objects} (or {\it nodes}), and
for all objects $c_1, c_2$, a set $Q(c_1, c_2)$, the members of which we call {\it morphisms} (or {\it arrows}) from $c_1$ to $c_2$.
\end{definition}

We may write $f : c_1 \to c_2$ or $c_1\To{f} c_2$ instead of $f \in C(c_1,c_2)$.

\begin{definition} For an arrow $f:c_1\to c_2$ in a quiver, we call
  $c_1$ the {\it source} of $f$ and $c_2$ the {\it target} of
  $f$.
\end{definition}

\begin{definition} In a quiver $Q$, a {\it path} from $c_1$ to $c_k$
  is a non-empty finite list of nodes and arrows
  $c_1\To{e_1}c_2\To{e_2}\cdots\To{e_{k-1}}c_k$.
\end{definition}

\begin{definition}\label{parallel} In a quiver $Q$, two paths from $c$ to $c'$ are called {\it parallel}.
\end{definition}

\begin{definition} A {\it category} $C$ is a quiver equipped with the
  following structure:
  \begin{itemize}
\item for all objects $c_1,c_2,c_3$, a function $\circ_{c_1,c_2,c_3} : C(c_2,c_3) \times C(c_1,c_2) \to C(c_1,c_3)$, which we call {\it composition}, and
\item for every object $c$, an arrow ${\sf id}_c \in C(c,c)$, which we call the {\it identity} for $c$. 
\end{itemize}
We may drop subscripts on ${\sf id}$ and $\circ$, when doing so does not create ambiguity.  These data  must obey axioms stating that $\circ$ is associative and ${\sf id}$ is its unit:
$$
{\sf id} \circ f = f \ \ \ \ f \circ {\sf id} = f \ \ \ \ f \circ (g \circ h) = (f \circ g) \circ h
$$
\end{definition}

\begin{definition} In a category $C$, the composition of a path
  $c_1\To{e_1}c_2\To{e_2}\cdots\To{e_{k-1}}c_k$ from $c_1$ to $c_k$ is
  defined recursively as ${\sf id}_{c_1}$ if $k=1$ and, if $k>1$, the composition of $e_1$ with the composition of the path $c_2\To{e_2}c_3\To{e_2}\cdots\To{e_{k-1}}c_k$.
\end{definition}

\begin{definition} A category $C$ is {\it small} if ${\sf Ob}(C)$ is a
  set and $C(c_1,c_2)$ is a set for all objects $c_1, c_2$.
\end{definition}

\begin{definition} Two morphisms $f : c_1 \to c_2$ and
  $g : c_2 \to c_1$ such that $f \circ g = {\sf id}$ and
  $g \circ f = {\sf id}$ are said to be an {\it isomorphism}.  We may
  also say in this situation that $f$ is an isomorphism.
\end{definition}
We write $c \in C$ to indicate $c \in {\sf Ob}(c)$ when it is clear that $c$ is an
object.


\begin{definition} An object $c$ of a category $C$ is called {\it initial} if for all $c'\in C$, there is a unique morphism $c\to c'$.  It is called {\it weakly initial} if for all $c'\in C$, there is a (not necessarily unique) morphism $c\to c'$.
\end{definition}
\begin{lemma}
All initial objects of a category are uniquely isomorphic (that is, for any two initial objects $c$ and $d$, there is exactly one isomorphism $c\to d$).  All weakly initial objects of a category are homomorphic (that is, for any two weakly initial objects $c$ and $d$, there is at least one morphism $c\to d$).
\end{lemma}

\begin{example}
The category ${\sf Set}$ has for objects all the sets in some set theory, such as ZFC, and for morphisms $X$ to $Y$ the (total, deterministic, not necessarily computable) functions
$X \to Y$.  The isomorphisms of ${\sf Set}$ are
exactly the bijections.
\end{example}
\begin{example}
A typed programming language gives a category,
with its types as objects and programs taking inputs of type $t_1$ and
returning outputs of type $t_2$ as morphisms $t_1 \to t_2$.  The composition of morphisms $f:t_1\to t_2$ and $g:t_2\to t_3$ is then defined as $(g\circ f)(x)\{ \textrm{return } g(f(x));\}$.
\end{example}


\begin{definition}
A {\it functor} $F : C \to D$ between categories $C$ and $D$ consists of:
\begin{itemize}
    \item a class function $F : {\sf Ob}(C) \to {\sf Ob}(D)$, and
    \item for every $c_1,c_2 \in {\sf Ob}(C)$, a function $F_{c_1,c_2} : C(c_1,c_2) \to D(F(c_1),F(c_2))$, where we may omit object subscripts when they can be inferred, such that 
\end{itemize}
$$
F({\sf id}_c) = {\sf id}_{F(c)} \ \ \ \ \ \ F(f \circ g) = F(f) \circ F(g).
$$
\end{definition}

\begin{example} The category ${\sf Cat}$ of all small categories, with functors as morphisms.
\end{example}

\begin{definition}
A {\it natural transformation} $h : F \to G$ between functors $F, G : C \to D$ consists of a family of morphisms $h_c : F(c) \to G(c)$, indexed by objects in $C$, called the {\it components} of $h$,
 such that for every $f : c_1 \to c_2$ in $C$ we have
$
h_{c_2} \circ F(f) = G(f) \circ h_{c_1}
$.
\end{definition}

The family of equations defining a natural transformation may be depicted as a {\it commutative diagram}: 
$$
\xymatrix{
F(c_1) \ar[d]_{h_{c_1}} \ar[r]^{F(f)} & F(c_2) \ar[d]^{h_{c_2}} \\
G(c_1) \ar[r]_{G(f)}& G(c_2) \\ 
}
$$
The commutativity of such a diagram means that any two parallel
paths in the diagram have the same composition in $D$; in this case,
the only non-trivial case is the two paths east-south and south-east.

\begin{example}
Given (small) categories $C$ and $D$, the functors from $C$ to $D$
form the functor category $D^C$, whose morphisms are
natural transformations.  
\end{example}

\begin{example} 
A relational database schema consisting of single-part foreign keys and
single-part unique identifiers also forms a category, say $C$, and we
may consider $C$-databases as functors $C \to {\sf Set}$~\cite{wadt},
with the pleasant property that natural transformations of such
functors correspond exactly to the $C$-database homomorphisms in the
sense of relational database
theory~\cite{Deutsch:2008:CR:1376916.1376938} (bearing in mind some
caveats alluded to in the introduction and discussed further in the conclusion to this paper and elsewhere).  Thus ${\sf Set}^C$ is the category of $C$-databases.
\end{example}

\begin{definition} A natural transformation is called a {\it natural isomorphism} when, considered as a morphism in a  category of functors and natural transformations, it is an isomorphism, or equivalently, when all of its components are isomorphisms. \end{definition}

Because categories are algebraic objects, they can be presented by
{\it generators and relations} (or as we like to say, generators and
equations) in a manner similar to e.g. groups~\cite{relfound}.

\begin{definition}
  The {\it free category} generated by a quiver $Q$ is the category
  ${\sf free}(Q)$ defined by
  \begin{itemize}
  \item ${\sf Ob}({\sf free}(Q))$ is defined as ${\sf Ob}(Q)$.
  \item for objects $c_1,c_2$, ${\sf free}(Q)(c_1,c_2)$ is defined as the set of all paths in $Q$ from $c_1$ to $c_2$.
  \item for paths $p=(c_1\To{e_1}\cdots\To{e_{k-1}}c_k)$ and $q=(c_k\To{e_k}\cdots\To{e_{l-1}}c_l)$, define $q\circ p$ as $c_1\To{e_1}\cdots\To{e_{l-1}}c_l$ (composition is path concatenation).
  \item for every object $c$, ${\sf id}_c$ is defined as the trivial path $c$.
  \end{itemize}
\end{definition}

\begin{definition}
  A {\it category presentation} $(Q,E)$ consists of a quiver $Q$
  and a set $E$ of pairs of parallel paths (see Definition~\ref{parallel}).  An element $(p,q)$ of
  $E$ is called an {\it path equation} and written $p=q$.
\end{definition}

\begin{definition}
  Let $C$ be a category.  A relation $\sim$ on the arrows of $C$ is
  called a {\it congruence} on $C$ if
  \begin{itemize}
  \item whenever $f\sim g$, $f$ and $g$ have the same source and target
  \item $\sim$ is an equivalence relation
  \item whenever $f,g:c\to d$ and $h:b\to c$, $f\sim g$ implies $f\circ h\sim g\circ h$
  \item whenever $f,g:c\to d$ and $k:d\to e$, $f\sim g$ implies $k\circ f\sim k\circ g$
  \end{itemize}
\end{definition}

\begin{definition}
  Given a category $C$ and a congruence $\sim$, the {\it quotient
    category} $C/\sim$ has the same objects as $C$ and its morphisms
  are $\sim$-classes of morphisms of $C$.  We define the source and
  target of a class $[f]$ as the source and target of $f$, define the
  identity of $c$ as $[{\sf id}_c]$, and define the composition 
  $[f]\circ [g]$ to be $[f\circ g]$.  Well-definedness follows from
  $\sim$ being a congruence.
\end{definition}

\begin{definition}
  In a category presentation $(Q,E)$, let $\sim_E$ be the
  smallest congruence on ${\sf free}(Q)$ containing $E$.  Then the
  category {\it presented by} $(Q,E)$ is the category
  ${\sf free}(Q)/\sim_E$.
\end{definition}

\begin{lemma}
  Let $(Q,E)$ be a category presentation.  Consider the following
  inference rules:

  \begin{alignat*}{3}
    \text{\sf Axiom }&\frac{(f=g)\in E}{f=g}& \text{\sf Ref }&\frac{}{f=f}& \text{\sf Sym }&\frac{f=g}{g=f} \\
    \text{\sf Trans }&\frac{f=g,\ g=h}{f=h}\quad  &\text{\sf RCong }&\frac{f=g:c\to d,\ h:b\to c}{f\circ h = g\circ h} &\quad\text{\sf LCong }&\frac{f=g:c\to d,\ k:d\to e}{k\circ f = k\circ g} \\
  \end{alignat*}

  If $f=g$ is provable in this calculus, we write $E\vdash f=g$.

  Then for morphisms $f,g:c\to d$, $f\sim_E g$ iff $E\vdash f=g$.
\end{lemma}


\begin{definition}
  Let $Q$ be a quiver.  A $Q${\it-algebra} $A$ consists of,
  \begin{itemize}
  \item for each object $c$ of $Q$, a set $Ac$, and
  \item for each morphism $f:c_1\to c_2$ of $Q$, a function $Af:Ac_1\to Ac_2$.
  \end{itemize}

  Given $Q$-algebras $A$ and $B$, a $Q${\it -algebra homomorphism}
  $\alpha$ is, for each object $c$ of $Q$, a function
  $\alpha_c:Ac\to Bc$ such that for all morphisms $f:c_1\to c_2$ of $Q$, the
  following diagram commutes:
  $$
  \xymatrix{
    Ac_1 \ar[d]_{\alpha_{c_1}} \ar[r]^{Af} & Ac_2 \ar[d]^{\alpha_{c_2}} \\
      Bc_1 \ar[r]_{Bf}& Bc_2 \\
    }
  $$
  To ease notation, in the following list, for a path $p=(c_0\To{f_1}c_1\To{f_2}\cdots\To{f_{n}}c_n)$, and operation $A$,  we write $Ap$ to indicate $Af_n\circ\cdots\circ Af_2\circ Af_1$.
\end{definition}
\begin{lemma}\label{functorpresentation}
  Let $(Q,E)$ be a category presentation.  Then:
  \begin{itemize}
  \item A functor from the category presented by $(Q,E)$ to
    ${\sf Set}$ is equivalent to a $Q$-algebra $A$ with $Ap=Aq$ whenever $(p=q)\in E$.
  \item A natural transformation between two such functors is
    equivalent to a $Q$-algebra homomorphism between the corresponding
    $Q$-algebras.
  \item A functor from the category presented by $(Q,E)$ to the
    category presented by $(Q',E')$ is equivalent to a morphism of
    presentations (``signature morphism'' in~\cite{relfound}), which we define inline as:
    \begin{itemize}
    \item For each object $c$ of $Q$, an object $Fc$ of $Q'$
    \item For each morphism $f:c_1\to c_2$ of $Q$, a path
      $Ff=(Fc_1\to c'\to\cdots\to Fc_2)$ of $Q'$.
    \end{itemize}
    such that for each equation $(c_0\To{f_1}c_1\To{f_2}\cdots\To{f_{n}}c_n)=(c_0\To{g_1}c'_1\To{g_2}\cdots\To{g_{n}}c_n)$ in $E$, we have that $E'\vdash Ff_n\circ\cdots\circ Ff_1 = Fg_n\circ\cdots\circ Fg_1$ ($\circ$ is composition in ${\sf free}(Q')$, i.e. concatenation).
  \end{itemize}
\end{lemma}
\begin{definition}
A {\it pushout} of objects $A,B,C$ and morphisms $f,g$ in a category, as shown below, is an object $D$ and morphisms $\alpha$ and $\beta$ as shown below, having the universal property that for any other such $D'$ and $\alpha'$ and $\beta'$, there is a unique morphism $\theta$ making the diagram commute:
\[
\xymatrix{
 A \ar[r]^{g}      \ar[d]_{f}         & C \ar[d]_{\beta} \ar@/^1pc/[rdd]^{\beta'}   \\
 B \ar[r]^{\alpha} \ar@/_1pc/[rrd]_{\alpha'} & D \ar@{-->}[rd]^{\theta} \pullbackcorner[ul] \\
 && D' 
}
\]

The dual notion of pushout is {\it pullback}.  Pushouts generalize to more complicated diagrams, in which case they are called colimits, but we do not define colimits here.

\end{definition}

\begin{definition}
  Given two functors $F : C \to D$ and $G : D \to C$, we say that $F$
  is {\it left adjoint} to $G$, written $F \dashv G$, when for every
  object $c$ in $C$ and $d$ in $D$ that the set of morphisms
  $F(c) \to d$ in $D$ is isomorphic to the set of morphisms
  $c \to G(d)$ in $C$, naturally in $c$ and $d$ (i.e., when we
  independently consider each side of the isomorphism as a functor
  $C \to {\sf Set}$ and as a functor ${\sf D} \to {\sf Set}$).
  
  Associated with each adjunction is a natural transformation $\eta : id_C \to G \circ F$ called the {\it unit} of the
  adjunction; a component $\eta_c : c\to G(F(c))$ of this transformation can be computed by applying the isomorphism $D(F(c),F(c))\cong C(c,G(F(c)))$ to the identity morphism ${\sf id}_{F(c)}$.
\end{definition}

\begin{definition}
  Given functors $F:C\to E$ and $G:D\to E$, the \textit{comma category} $F/G$ is a category whose objects are triples $(c\in C, d\in D, f:Fc\to Gd)$ and whose morphisms are pairs $(h:c\to c',k:d\to d'): (c,d,f)\to (c',d',f')$ such that the following square commutes:
  \[
\xymatrix{
 Fc \ar[r]_{f} \ar[d]_{Fh} & Gd \ar[d]^{Gk} \\
 Fc' \ar[r]^{f'} & Gd' \\
}
\]
\end{definition}

In this paper, almost all of the comma categories we will consider are of the form $e/G$, where $e$ is the inclusion of an object $e\in E$.  Then objects are simply pairs $(d\in D,f:e\to Gd)$ and morphisms are $(k:d\to d'):(d,f)\to (d',f')$ such that the following triangle commutes:
\[
\xymatrix{
  & e \ar[dl]_f \ar[dr]^{f'} \\
  Gd\ar[rr]_{Gk} && Gd'
}
\]

\begin{lemma} \label{lemma:adjunction}
  Suppose that a functor $G:D\to C$ has a left adjoint $F$.  Then for all $c\in C$, $(Fc,\eta_c)$ is an initial object in the comma category $c/G$.  Explicitly, for any $d\in D$ and $f:c\to Gd$, there is a unique $k:Fc\to d$ such that the following diagram commutes:
\[
\xymatrix{
  & c \ar[dl]_{\eta_c} \ar[dr]^{f} \\
  FGc\ar[rr]_{Fk} && Fd
}
\]

Conversely, if, for a functor $G:D\to C$, the category $c/G$ has an initial object $(Fc, \eta_c)$ for all $c$, then $F:{\sf Ob}(C)\to {\sf Ob}(D)$ extends to a functor $C\to D$ which is left adjoint to $G$ and has unit $\eta$.
\end{lemma}

\begin{example}\label{leftkan}
The most important example of an adjunction in this paper is a left Kan extension.  Let $F:C\to D$ be a functor and consider the functor $\Delta_F:{\sf Set}^D\to {\sf Set}^C$ defined by pre-composition with $F$: $\Delta_F(I)\coloneqq I\circ F$, and for $\alpha:I\to I'$, $(\Delta_F(\alpha))_{d} \coloneqq \alpha_{Fd}$.  Whenever $C$ is small, this functor has a left adjoint $\Sigma_FI:{\sf Set}^C \to {\sf Set}^D$, called the {\it left Kan extension} of $I$ by $F$.  Applying the previous lemma, we find that for a $C$-database $I$, the object $(\Sigma_F(I),\eta_I)$ of $I/\Delta_F$ is initial (where $\eta_I$ is the unit of the $\Sigma_F \dashv \Delta_F$ adjunction).  We give a formula for $\Sigma_F(I)$ (see~\cite{riehl}):
\begin{equation}\label{leftkanformula}
\Sigma_F(I)(d)=\colim(F/d\To{\Pi} C\To{I} {\sf Set})
\end{equation}
where $\Pi$ is the canonical projection functor.
\end{example}
\begin{example}
Another important example of an adjunction is an inclusion $U: C\hookrightarrow D$ which has a left adjoint $R:D\rightarrow C$.  Then $C$ is called a {\it reflective subcategory} of $D$ and $R$ is called the ${\it reflector}$.
\end{example}

\section{Left Kan Extensions as Free Models of Cartesian Theories}
\label{lk}

In this section we show that left Kan Extensions can be considered as free models.  
We:

\begin{itemize}
    \item define regular and cartesian logic (Section~\ref{flt}), as well as the {\it cartesian theory of a category}; and,
    \item define free and weakly free models of a theory on an input instance (Section~\ref{initweaklyinit}); and,
    \item define the {\it cograph} of a functor, a category ${\sf cog}(F)$ (Section~\ref{sec.cograph}); and,
    \item show that the left Kan extension of a $C$-instance $I$ by a functor $F:C\to D$ is a free model of the cartesian theory of ${\sf cog}(F)$ on $I$, considered as an input instance (section~\ref{prf}).
\end{itemize}



We begin by describing a running example left Kan computation.

\subsection{Running Example of a Left Kan Extension}
\label{sec:ex}
Our running example of a left Kan extension is that of quotienting a set by an equivalence relation, where the equivalence relation is induced by two given functions.  In this example, the input data consists of teaching assistants ({\sf TA}s), {\sf Faculty}, and {\sf Student}s, such that every TA is exactly one faculty and exactly one student.  We wish to compute all of the persons without double-counting the TAs, which we can do by taking the disjoint union of the faculty and the students and then equating the two occurrences of each TA.  

Our source category $C$ is the category {\sf Faculty'} $\leftarrow$ {\sf TA'} $\rightarrow$ {\sf Student'}, our target category $D$ extends $C$ into a commutative square with new object, {\sf Person} with no ${\sf '}$ marks for disambiguation, and our functor $F:C \to D$ is the inclusion:

\[
C := \parbox{1.63in}{\fbox{\xymatrix@=8pt{
& \LTO{TA'} \ar[dl]_{\sf isTF'} \ar[dr]^{\sf isTS'} & \\
\DTO{Faculty'}  & & \DTO{Student'} \\
& \DTOW{Person'} & \\ }}}
\Too{F}  
\parbox{1.67in}{\fbox{\xymatrix@=8pt{
& \LTO{TA} \ar[dl]_{\sf isTF} \ar[dr]^{\sf isTS} & \\
\DTO{Faculty}  \ar[dr]_{\sf isFP} & {}^{ {\sf isTF}.{\sf isFP} = {\sf isTS}.{\sf isSP}} & \DTO{Student}  \ar[dl]^{\sf isSP} \\
& \DTO{Person} & }}} \ \ \ \ \ \ \ \ \ \ \ \ \ =: D
\]
Our input functor $I : C \to \Set$, displayed with one table per object, is:
\[
\begin{tabular}{ >{\sffamily}l}
 Faculty'    \\\hline 
 'Dr.' Alice  \\
 'Dr.' Bob   \\
 Prof. Ed  \\
 Prof. Finn  \\
 Prof. Gil  
\end{tabular}
\hspace{.5in}
\begin{tabular}{>{\sffamily}l}
 Student'     \\\hline 
Alice  \\
Bob  \\
Chad \\
Doug \\ \\
\end{tabular}
\hspace{.5in}
\begin{tabular}{>{\sffamily}l>{\sffamily}l>{\sffamily}l}
  TA'  &  isTF'  & isTS'  \\\hline 
math-TA &'Dr.' Alice&Alice\\
cs-TA  &'Dr.' Bob&Bob\\ \\ \\ \\
\end{tabular}
\]
The {\sf cs-TA} is both {\sf 'Dr.' Bob} and {\sf Bob}, and the left Kan extension equates them as persons.  Similarly, the {\sf math-TA} is both {\sf 'Dr.' Alice} and {\sf Alice}.  We thus expect $5+4-2=7$ persons in $\Sigma_F(I)$.  However, there are infinitely many left Kan extensions $\Sigma_F(I)$; each is naturally isomorphic to the one below in a unique way.  That is, the following tables uniquely define $\Sigma_F(I)$ up to choice of names:
\[
\begin{tabular}{>{\sffamily}l>{\sffamily}l}
 Faculty & isFP     \\\hline 
  'Dr.' Alice  & math-TA \\ 
  'Dr.' Bob  &  cs-TA  \\ 
  Prof. Ed  & Prof. Ed \\ 
   Prof. Finn  & Prof. Finn \\ 
   Prof. Gil  & Prof. Gil \\
  \\ \
\end{tabular}
\hspace{.15in}
\begin{tabular}{>{\sffamily}l>{\sffamily}l}
 Student & isSP    \\\hline 
Alice  & math-TA \\ 
Bob  & cs-TA  \\ 
Chad  &   Chad \\ 
 Doug  & Doug \\ \\ \\ \\ 
\end{tabular}
\hspace{.15in}
\begin{tabular}{>{\sffamily}l>{\sffamily}l>{\sffamily}l}
  TA  &  isTF & isTS  \\\hline 
math-TA &'Dr.' Alice&Alice\\
cs-TA  &'Dr.' Bob&Bob\\ \\ \\ \\ \\ \\
\end{tabular}
\hspace{.15in}
\begin{tabular}{>{\sffamily}l>{\sffamily}l}
  Person \\\hline 
  Chad \\
   cs-TA  \\ 
 Doug  \\ 
   Prof. Ed \\ 
  Prof. Finn \\ 
   Prof. Gil \\   
    math-TA \\  
\end{tabular}
\]

In this example the natural transformation $\eta_I:I\to \Delta_F(\Sigma_F(I))$, i.e.\ the $I$-component of the unit of the $\Sigma_F \dashv \Delta_F$ adjunction, is an isomorphism of $C$-instances; it associates each source {\sf Faculty'} to the similarly-named target {\sf Faculty}, etc.  This is not generally the case; the reason it is true here is that $F$ is fully faithful, so for $c\in C$ we have $F/Fc\cong {\sf id}_C/c$, and the colimit formula \ref{leftkanformula} gives
\[
\Delta_F\Sigma_FI(c)=\colim({\sf id}_C/c\To{\Pi} C\To{I} {\sf Set})\cong c
\]

\subsection{Regular and Cartesian Theories and Models}\label{flt}
\begin{definition}
  A {\it signature} $\sigma$ consists of a set $S$ of sorts and a set $R$ of relation symbols, each with a {\it sorted arity} -- that is, a list of sorts.
  
  An {\it instance} $I$ on  $\sigma$, also called a $\sigma$-{\it instance}, consists of a set $Is$ for each sort $s\in S$ and a relation $Ir\subseteq Is_0\times\cdots\times Is_n$ for each relation symbol $r\in R$ of arity $s_0,\ldots,s_n$.  
  We call any $v \in \bigcup_{s\in S}Is$ an {\it element} of $I$; we try to use sans serif names for instance elements, and the letters $u$, $v$, $w$ for variables ranging over instance elements.  
  
  A instance $J$ is a {\it subinstance} of $I$ if $Js\subseteq Is$ for each sort $s\in S$ and $Jr\subseteq Ir$ for each relation symbol $r\in R$.
  
  A {\it morphism of $\sigma$-instances} $f:I\to J$ is a sort-indexed family of functions $f_s:Is\to Js$ such that for every relation symbol $r\in R$ of arity $s_0,\ldots,s_n$ and all elements $v_0\in Is_0,\ldots,v_n\in Is_n$ we have that $r(v_0,\ldots,v_n)$ implies $r(f_0(v_0),\ldots,f_n(v_n))$.  When it is clear from context, we omit subscripts on these functions.
  
  A morphism of $\sigma$-instances can equivalently be described as a sort-indexed family of functions $f_s:Is\to Js$ such that for every relation symbol $r\in R$ of arity $s_0,\ldots,s_n$, the east-south path in the following diagram factors through $Jr$, as shown by the dashed arrow:
  
  \[ \begin{tikzcd}
  Ir \ar[d,dashrightarrow, "f_r"] \ar[r,hook] & Is_0\times\cdots\times Is_n \ar[d, "f_{s_0}\times\cdots\times f_{s_n}"] \\
  Jr \ar[r,hook] & Js_0\times\cdots\times Js_n
  \end{tikzcd} \]
  
  A morphism of $\sigma$-instance is called {\it surjective} if all components $f_s$ are surjective and all induced components $f_r$ are also surjective. 
  
  Instances and morphisms of $\sigma$-instances form a category, $\sigma{\sf -Inst}$.
\end{definition}

The following lemma provides another perspective on instances.
\begin{lemma}\label{instancesascopresheafs}
Let $C$ be the free category on the quiver with objects $S\sqcup R$ and a morphism $p_{n}^r:r\to s$ whenever the $n$th sort in the arity of $r$ is $s$.  Then $\sigma{\sf -Inst}$ embeds as a full subcategory of ${\sf Set}^C$ via the mapping which sends an instance $I$ on $\sigma$ to a functor sending $s\in S$ to $Is$, sending $r\in R$ to $Ir$, and sending $p_{n}^r$ to the projection of $Ir$ onto its $n$th component.
\end{lemma}

We next discuss syntax.  We leave many of the details informal, but see \cite{johnstone_2002} for a fuller treatment.  We try to use the letters $x,y,z$ (possibly with subscripts or primes) as (object language) variables, and we assume all variables have unique \textit{sort}s, writing $x:s$ to denote that variable $x$ has sort $s$.


\begin{definition}
Given a signature $\sigma=(S,R)$, a {\it regular formula} is a (possibly empty, indicating truth) conjunction of 
\begin{itemize}
\item {\it Equational atoms}:  assertions $x=y$, where $x$ and $y$ have the same sort in $S$, and
\item {\it Relational atoms}:  assertions $r(x_0,\ldots,x_n)$, where the sorts of $x_0,\ldots,x_n$ correspond to the arity of $r\in R$.
\end{itemize}
\end{definition}

\begin{definition}
Given a signature $\sigma$, let $\phi(x_0,\ldots,x_n)$ be a regular formula.  Let $I$ be a $\sigma$-instance.  We then define the $I$-interpretation of $\phi$ as the relation $I\phi\subseteq Is_0\times\cdots\times Is_n$ defined as
\begin{itemize}
    \item If $\phi(x_0,\ldots,x_n)=(x_i=x_j)$, then $I\phi\coloneqq\{(x_0,\ldots,x_n)\in Is_0\times\cdots\times Is_n \mid x_i=x_j\}$.
    \item If $\phi(x_0,\ldots,x_n)=r(x_{i_0}, \ldots, x_{i_k})$, then $I\phi\coloneqq\{(x_0,\ldots,x_n)\in Is_0\times\cdots\times Is_n \mid Ir(x_{i_0}, \ldots, x_{i_k})\}$.
    \item If $\phi(x_0,\ldots,x_n)=\bigwedge_j\phi_j(x_0,\ldots,x_n)$, then $I\phi\coloneqq\bigcap_jI\phi_j$.
\end{itemize}
\end{definition}

\begin{definition}
Given a signature $\sigma$, an {\it embedded dependency} (ED), or {\it regular sequent} $\xi$ is a constraint of the form
\begin{equation}\label{ED}
  \forall (x_0:s_0) \cdots  (x_n:s_n) \ldotp
  \phi(x_0 , \ldots , x_n) \Rightarrow
  \exists (x_{n+1}:s_{n+1}) \cdots \ (x_m:s_m)\ldotp \psi(x_0, \ldots, x_m)
\end{equation}
where $\phi$ and $\psi$ are regular formulas.

If $\psi$ consists only of equational atoms, $\xi$ is called an {\it equality-generating dependency} (egd).  If $\psi$ consists only of relational atoms, $\xi$ is called an {\it tuple-generating dependency} (tgd).
\end{definition}

The following is straightforward.

\begin{lemma}\label{EDisegdsandtgds}
Every ED is logically equivalent to a conjunction of egds and tgds.
\end{lemma}
\begin{proof}
See~\cite{abiteboul_hull_vianu_1996}.
\qed \end{proof}

\begin{definition}
A {\it regular theory}~\cite{johnstone_2002} $\ax$
on a signature $\sigma$ is a set of EDs on $\sigma$.  If an instance of $\sigma$ satisfies all of the EDs in $\ax$ in the usual way, it is called a ${\it model}$ of $\ax$.  A {\it morphism of models} of $\ax$ is defined as a morphism of $\sigma$-instances.  Models and morphisms of models of $\ax$ form a category, ${\sf Mod}(\ax)$.
\end{definition}

\begin{definition}
A {\it cartesian theory}~\cite{johnstone_2002}\footnote{This definition is not the same as that given in \cite{johnstone_2002}, but it is equivalent.} on a signature $\sigma$ is a set of constraints of the form
\begin{equation}\label{cartED}
  \forall (x_0:s_0) \cdots  (x_n:s_n) \ldotp
  \phi(x_0 , \ldots , x_n) \Rightarrow
  \exists ! (x_{n+1}:s_{n+1}) \cdots \ (x_m:s_m)\ldotp \psi(x_0, \ldots, x_m)
\end{equation}
where $\phi$ and $\psi$ are as before, and $\exists !$ means ``exists unique.''
\end{definition}

\begin{lemma}\label{cartisreg}
  Every cartesian theory is logically equivalent to a regular theory.
\end{lemma}
\begin{proof}
We can rewrite (\ref{cartED}) as the conjunction of (\ref{ED}) and
\begin{multline}
  \forall (x_0:s_0) \cdots  (x_n:s_n) (x_{n+1},x_{n+1}':s_{n+1})\cdots (x_m,x_m':s_m) \ldotp
  \phi(x_0 , \ldots , x_n)\ \wedge\\ \psi(x_{n+1},\ldots,x_m)\ \wedge\ \psi(x_{n+1}',\ldots,x_m')  \Rightarrow x_{n+1}=x_{n+1}'\ \wedge\ \cdots\ \wedge\ x_m=x_m'
\end{multline}
\qed\end{proof}


Cartesian theories can be equivalently described
\begin{itemize}
    \item using spans instead of relations, in which case they are called finite-limit sketches~\cite{johnstone_2002}
    \item using partial functions instead of relations, in which case they are often called {\it essentially algebraic theories}~\cite{nlab:essentially_algebraic_theory}.  Also see \cite{PALMGREN2007314}.
\end{itemize}


\subsubsection*{The Cartesian Theory of a Category Presentation}\label{sec.theoryofcat}
We now describe how to convert a presentation $(Q,E)$ of a category $C$ 
into a cartesian theory $\ax$ that axiomatizes the functors $C \to {\sf Set}$. To do so, first let $\sigma$ be the signature whose sorts are the objects of $Q$ and whose relation symbols are the morphisms $f:c\to d$ of $Q$, assigned the arity $c,d$.  Then let the theory $\ax$ be comprised of
\begin{itemize}
    \item axioms requiring all relations $f:c\to d$ be total and functional:
\[
 \forall (x : c)\ldotp \ \exists! (y : d)\ldotp f(x,y).
\]
\item and, for each equation\\ $(c\To{f_0}d_1\To{f_1}\cdots\To{f_{m-1}}d_m\To{f_m}c')=(c\To{g_0}e_1\To{g_1}\cdots\To{g_{n-1}}e_n\To{g_n}c')$ in $E$, an axiom
\[
 f_0(x,y_0)\wedge\cdots\wedge f_m(y_{m-1},y_m)\wedge g_0(x,z_0)\wedge\cdots\wedge g_n(z_{n-1},z_n) \Rightarrow y_m=z_n
\]
\end{itemize}

(Note that from now on, we will omit universal quantifiers and sorts when they can be inferred from context, but we will continue to make existential quantifiers explicit, as above.)

\begin{lemma}\label{theoryofcat}
  The categories ${\sf Set}^C$ and ${\sf Mod}(\ax)$ are isomorphic.
\end{lemma}

To refer to the theory of a presentation of $C$, we will sometimes say metonymically ``the theory C''.

\subsection{Free and Weakly Free Models of Theories on Given Instances}\label{initweaklyinit}

\begin{definition} Let $\sigma$ be a signature and $\ax$ be a regular theory on $\sigma$.  Let $U:{\sf Mod}(\ax)\to \sigma{\sf -Inst}$ be the forgetful functor.  Let $I$ be a instance on $\sigma$. Then a (weakly) free model of $\ax$ on $I$ is a model $A$ of $\ax$ and a morphism $h:I\to A$ such that for any model $A'$ of $\ax$ and morphism $h':I\to A'$, there is a (not necessarily) unique morphism $g:A\to A'$ such that $g\circ h = h'$.
\end{definition}

Equivalently, a {\it (weakly) free model of} $\ax$ {\it on the input instance} $I$ is a (weakly) initial object of $I/U$.  

The database theory literature considers {\it universal models}, an even weaker notion than weakly free models.  

\begin{definition}
A universal model of $\ax$ on $I$ is a model $A$ of $\ax$ and a morphism $h:I\to A$ such that for any model $A'$ of $\ax$ and morphism $h':I\to A'$, there is a morphism $g:A\to A'$.\footnote{Our definition differs from that in \cite{Deutsch:2008:CR:1376916.1376938} in that we do not mandate finiteness.}
\end{definition}

Notice that the notion of universal model generalizes that of weakly free model by leaving off the final criterion.  The following lemma relates the notions of universal model and weakly free model.

\begin{lemma}\label{universalvsweaklyinitial}
Let $\ax$ be a regular theory on a signature $\sigma=(S,R)$ and $I$ a given $\sigma$-instance.
\begin{enumerate}
    \item If $I=\varnothing$, then $I/U\cong {\sf Mod}(\ax)$, so weakly free models of $\ax$ on $I$ are exactly universal models of $\ax$ on $I$.
    \item Define the signature $\sigma'=\sigma\sqcup \{p_v\mid v \textrm{ is an element of }I \}$, where $p_v$ is given arity $s$ whenever $v\in Is$.  Define the theory $\ax'=\ax\sqcup \{\exists! x\ldotp p_v(x)\mid v \textrm{ is an element of }I \} \sqcup
    \{p_{v_1}(x_1)\wedge\cdots\wedge p_{v_n}(x_n)\Rightarrow r(x_1,\ldots,x_n)\mid Ir(v_1,\ldots,v_n) \}$.  Then the category $I/U$ is isomorphic to the category ${\sf Mod}(\ax')$, so weakly free models of $\ax$ on $I$ are exactly weakly free models of $\ax'$ on $\varnothing$, which by part $1$ are exactly universal models of $\ax'$ on $\varnothing$.
    \item For this paragraph, locally introduce the notion of ``constant'', as is standard in database theory~\cite{Deutsch:2008:CR:1376916.1376938}, as follows.  We first (locally) require all elements of all instances we will ever consider to be drawn from either from a fixed universe $C$ of ``constants'' or a fixed universe $V$ of ``labelled nulls'', where $C$ and $V$ are disjoint; and we then extend the definition of morphisms of instances $h:I\to I'$ to mandate that $h(c)=c$ whenever $c$ is a constant.  If our instance $I$ is entirely comprised of constants, then weakly free models of $\ax$ on $I$ are exactly universal models of $\ax$ on $I$.
\end{enumerate}
\end{lemma}

One might hope that the construction in the third part of this lemma reduces the notion of ``weakly initial model'' to that of ``universal model'': just replace all elements in your instance with ``constants'' in the sense of the above.  However, the following example shows that this is not the case:

\begin{example} \textbf{Replacing all instance elements with constants (in the sense of Lemma~\ref{universalvsweaklyinitial} part 3) may invalidate previously valid universal models.}
Let $I=\{{\sf foo},{\sf bar} \}$ and $\ax=\{\forall x,y.x=y\}$.  Then we clearly have the weakly free (and thus universal) model $(\{{\sf baz}\}, !)$ where $!:I\to\{{\sf baz}\}$ is the unique morphism.  But if we replace all elements in $I$ with constants (in the sense of the above) to obtain $I'=\{c,d\}$, then there is no longer any universal model of $\ax$ over $I'$.  If there was a universal model $(A,a)$, we would have $a(c)=a(d)$, which is not possible since $c$ and $d$ are distinct constants.  So yes, as per Lemma~\ref{universalvsweaklyinitial}, weakly initial models of $\ax$ over $I'$ are exactly universal models of $\ax$ over $I'$, but in this example this equivalence does not matter, as there are none of either.
\end{example}

Also see Lemma~\ref{dataexchange} for more on the comparison between the notions of ``weakly free model'' and ``universal model'', and how they interact with the database theoretic distinction of ``constants'' and ``labelled nulls''.


Every free model is weakly free, and every weakly free model is universal.  The following example shows that the finite universal models exist more frequently than finite weakly free models (we will see in Corollary~\ref{existence} that weakly free models always exist).

\begin{example} \label{universalcounterexample}
\textbf{Existence of a finite universal model and a weakly free model does not imply existence of a finite weakly free model, even on a cartesian theory.}
Take the uni-typed signature $\sigma$ with a single binary relation $r$, and let $\ax=\{\forall x \exists! y \ldotp r(x,y), \exists! x\ldotp r(x,x) \}$ and $I=\{{\sf foo}\}$. Then the model $U=\{{\sf foo}\}$, $Ur=\{({\sf foo},{\sf foo})\}$ and the unique morphism $!:I\to U$ form a finite universal model of $\ax$ on $I$.  The model $N$ whose elements are the natural numbers and where $Nr=\{(0,0)\}\cup\{(m,n)\mid m+1=n>1 \}$, and the morphism $f:I\to N$, $f({\sf foo})=1$ form a weakly free model (in fact, a free model) of $\ax$ on $I$.

But suppose that $(A,a:I\to A)$ is a finite weakly free model of $\ax$ on $I$.  Consider again the model $N$ and the morphism $f:I\to N$, $f({\sf foo})=1$.  Then there must be a morphism $g:A\to N$ such that $g\circ a=f$, i.e. $g(a({\sf foo}))=f({\sf foo})=1$.  Since $A$ is finite, let $k\geq 1$ be the largest number such that $g^{-1}(k)$ is nonempty, and let $u\in g^{-1}(k)$.  Then there is a $v\in A$ such that $Ar(u,v)$, so $Nr(k,g(v))$, so $g(v)=k+1$, a contradiction.
\end{example}

However, the same is not true of the inclusion (free models $\subseteq$ weakly free models). 

\begin{lemma}\label{finiteweaklyinitial}
Existence of a finite weakly free model implies that every free model is finite.
\end{lemma}
\begin{proof}
Let $(A,a)$ be a free model and $(B,b)$ be a finite weakly free model of $\ax$ on $I$.  Then there exist morphisms $f:A\to B$ and $g:B\to A$ such that $f\circ a =b$ and $g\circ b = a$.  Then $g\circ f\circ a = a = {\sf id}_A \circ a$, so by the definition of ``free model'', $g\circ f = {\sf id}_A$.  Thus $B$ surjects onto $A$, so $A$ must be finite.
\qed \end{proof}

It may be tempting to consider the 4th possible definition: we say a {\it strictly universal model} of $\ax$ on $I$ is a model $A$ of $\ax$ and a morphism $h:I\to A$ such that for any model $A'$ of $\ax$ and morphism $h':I\to A'$, there is a unique morphism $g:A\to A'$.  However, without the last equation, uniqueness becomes very hard to guarantee, even in the simplest cases.  For example, take the uni-typed signature with no relation symbols, $\ax=\varnothing$, and $I=\{{\sf foo}\}$.  Then $(I,{\sf id}_I)$ is a free model of $\ax$ on $I$, but not a strictly universal model of $\ax$ on $I$, because $A=\{{\sf foo},{\sf bar}\}$ is a model of $\ax$ and there are multiple morphisms $I\to A$.  It only makes sense to consider strictly universal models of $\ax$ on $I$ when $I=\varnothing$, and in this case, they are exactly free models.


As we will see in Corollary~\ref{existence}, for any (possibly infinite) regular theory $\ax$ on a (possibly infinite) signature $\sigma$ and for any (possibly infinite) $\sigma$-instance $I$, there exists a (possibly infinite) weakly free model of $\ax$ on $\sigma$, and, if $\ax$ is cartesian, there exists a (possibly infinite) free model of $\ax$ on $\sigma$.


\subsection{The Cograph of a Functor}\label{sec.cograph}

\begin{definition}
  The {\it cograph}~\cite{cograph} of a functor $F : C \to D$ is a category, written ${\sf cog}(F)$,  presented as follows: 
  we first take the co-product of $C$ and $D$ as categories (i.e., take the disjoint union of $C$ and $D$'s objects, generating morphisms, and equations), we then add a generating morphism $\alpha_c : c \to F(c)$ for each object $c \in C$, and finally we add an equation $F(f)\circ\alpha_c= \alpha_{c'}\circ f $ for each generating morphism $f : c \to c' \in C$.
\end{definition}

Categorically, the cograph of $F$ is the {\it collage}~\cite{GARNER20161} of the profunctor $D(F-,-)$ represented by $F$.  For example:

$$
\parbox{4.0in}{\xymatrix@=8pt{
& \LTO{TA'} \ar[rrrr]^{\alpha} \ar[dl]_{\sf isTF'} \ar[dr]^{\sf isTS'} & & & &  \LTO{TA} \ar[dl]_{\sf isTF} \ar[dr]^{\sf isTS}  \\
\DTO{Faculty'}  \ar@/_1.5pc/[rrrr]^{\alpha}& & \DTO{Student'} \ar@/_3.1pc/[rrrr]^{\alpha} &  &  \DTO{Faculty} \ar[dr]_{\sf isFP} & {}^{ {\sf isFP}\circ {\sf isTF} = {\sf isSP}\circ {\sf isTS} }  & \DTO{Student}   \ar[dl]^{\sf isSP}     \\ 
 & {}^{ \alpha\circ {\sf isTF'} = {\sf isTF}\circ \alpha}   &   {}^{ \alpha\circ {\sf isTS'} = {\sf isTS}\circ \alpha} & & &  \DTO{Person} }} 
 $$

 The evident inclusion functors $i_C : C \to {\sf cog}(F)$ and $i_D : D \to {\sf cog}(F)$ of $C$ and $D$ into ${\sf cog}(F)$ will be used several times throughout the paper. The following proposition characterizes ${\sf cog}(F)$-instances.

\begin{proposition}\label{prop:cograph_cat}
Let  $F\colon C\to D$ be a functor. The following are equivalent:
\begin{enumerate}
	\item the category ${\sf id}_{{\sf Set}^C}/\Delta_F$ of triples $(I,J,f)$, with $I: C\to\Set$, $J:D\to\Set$, and $f\colon I\to\Delta_F(J)$ (where a morphism $(I,J,f) \to (I',J',f')$ is a pair of morphisms $i:I\to I'$ and $j:J\to J'$ such that $\Delta_F(j) \circ f = f' \circ i$),
	\item the category $\Sigma_F/{\sf id}_{{\sf Set}^D}$ of triples $(I,J,f)$, with $I: C\to\Set$, $J:D\to\Set$, and $f\colon \Sigma_F(I)\to J$,
	\item the category ${\sf Set}^{{\sf cog}(F)}$ of functors ${\sf cog}(F)\to\Set$.
\end{enumerate}
\end{proposition}
\begin{proof}
$1\leftrightarrow 2$ is the definition of $\Sigma_F$ being left adjoint to $\Delta_F$, naturally in $I, J$. 
$3 \to 1$. Given a functor $K:{\sf cog}(F)\to {\sf Set}$, we compose it with $i_C: C\to {\sf cog}(F)$ to obtain a functor $I:=K\circ i_C : C\to {\sf Set}$, and similarly we obtain $J:=K\circ i_D: D\to {\sf Set}$. To give a natural transformation $f: I\to \Delta_F(J)$, first choose an object $c$ in $C$. We need a function $I(c)\to J(F(c))$, so we use $K(\alpha_c): I(c)=K(i_C(c))\to K(i_D(F(c))=J(F(c))$. For any morphism $g: c \to c'$, the corresponding equation $F(g)\circ\alpha_c = \alpha_{c'}\circ g$ in ${\sf cog}(F)$ ensures that $f$ is indeed natural. This establishes $3\to 1$ on objects, and it is straightforward to check that it functorial.
$1 \to 3$. As expected, this is just inverse to the above. Given $I, J$, and $f$, we define $K: {\sf cog}(F)\to Set$ on objects via $I$ and $J$ on objects. Every generating morphism in ${\sf cog}(F)$ is either in $C$ or in $D$---in which case use $I$ or $J$---or it is is of the form $\alpha_c: c \to F(c)$, in which case use $f_c: I(c) \to J(F(c)).$ The equations in ${\sf cog}(F)$ are satisfied by the naturality of $f$. This establishes $1\to 3$ on objects, and it is again straightforward to check that it is functorial.
\qed\end{proof}

\subsection{Left Kan Extensions Using Free Models}\label{prf}

To compute the left Kan extension $\Sigma_F(I)$ of $I:C\to {\sf Set}$ along $F:C \to D$ using the previous lemma, we consider $I$ as an instance $\mathcal{I}$ on the signature $\sigma$ of the cartesian theory ${\sf cog}(F)$, compute $\init_{{\sf cog}(F)}(\mathcal{I})$, and then project the $D$ part we need for $\Sigma_F(I)$.  Our main result is:


\begin{lemma} \label{quadruple}
Define a $\sigma$-instance $\mathcal{I}$ by setting $\mathcal{I}c=Ic$ for $c\in C$, setting $\mathcal{I}f=If$ for $f:c\to c'$ in $C$, and setting $\mathcal{I}o=\varnothing$ for all other sorts and relation symbols $o$.

Consider the free model $(\init_{{\sf cog}(F)}(\mathcal{I}),h)$ of ${\sf cog}(F)$ on $\mathcal{I}$.  Discarding $h$, Lemma~\ref{theoryofcat} allows us to consider $\init_{{\sf cog}(F)}(\mathcal{I})$ as a functor ${\sf cog}(F)\to {\sf Set}$.  By $1\leftrightarrow 3$ in Proposition~\ref{prop:cograph_cat}, this functor in turn can be considered as a triple $(I':C\to{\sf Set},J':D\to{\sf Set},f':I'\to\Delta_F(J'))$.

Thus considering $\init_{{\sf cog}(F)}(\mathcal{I})$ as a triple, we have that $\init_{{\sf cog}(F)}(\mathcal{I})\cong(I,\Sigma_F(I),\eta_I)$, where $\eta$ is the unit of the $\Sigma_F \dashv \Delta_F$ adjunction.
\end{lemma}
\begin{proof}
The LHS is the initial object of $\mathcal{I}/{\sf Mod}({\sf cog}(F))$, which upon inspection is seen to be isomorphic to $I/\Delta_{i_C}$, whose objects are quadruples are $(I':C\to{\sf Set},J':D\to{\sf Set},f':I'\to\Delta_F{J'},u':I\to I')$ and whose morphisms are pairs $(i,j)$ such that the following diagram commutes:
\[
\xymatrix @R=.5pc{
  & I' \ar[r]^{f'} \ar[dd]^{i} & \Delta_F{J'} \ar[dd]^{\Delta_F(j)} \\
  I \ar[ru]^{u'} \ar[rd]_{u''} \\
  & I'' \ar[r]^{f''} & \Delta_F{J''}
}
\]

We show that $(I, \Sigma_F(I), \eta_I, {\sf id}_I)$ is the initial object of this category.  Given a quadruple $(I', J', f', u')$, consider the diagram
\[
\xymatrix @R=.5pc{
  & I \ar[r]^-{\eta_I} \ar[dd]^{i} & \Delta_F\Sigma_F(I) \ar[dd]^{\Delta_F(j)} \\
  I \ar@{=}[ru] \ar[rd]_{u'} \\
  & I' \ar[r]^{f'} & \Delta_F{J'}
}
\]

For this diagram to commute, we must have $i=u'$, which reduced the diagram to the square.  By the definition of left Kan extension, there is a unique $j$ making this square commute, so the theorem follows.
\qed \end{proof}

Notice that this lemma does not require $C$, $D$, $F$, or $I$ to be finitely-presented.  However, if they are not, the resulting theory ${\sf cog}(F)$, its signature $\sigma$, and the instance $\mathcal{I}$ will not all be finite, which will impede computation.  The lemma is nonetheless given in full generality, as it is of mathematical interest that all left Kan extensions can be considered as free models of cartesian theories, even ones that are not computable.

With the above lemma in hand, we now turn to using chase algorithms for computing free models of cartesian theories.

\section{Chases on Cartesian Theories}\label{chasecart}

Chases are a class of algorithms used for computing weakly free models of regular theories on input instances.  A run of a chase algorithm gives rise to a sequence of instances, with the sequence itself called a chase sequence.  In this section we discuss the application of these algorithms to cartesian theories.  We discuss the {\it standard} chase and {\it parallel} chase in this section.  In the Appendix we discuss the {\it core} chase, and we define a slight variant called the {\it categorical core chase} which works better with the language of weakly free models than universal models.

\subsection{Summary of New Results}


If $\ax$ is a regular theory, then the parallel chase and the categorical core chase both compute finite weakly free models of $\ax$ on a finite input instance.  Moreover, the categorical core chase is complete (see Lemmas~\ref{chasecomputesweaklyinit},~\ref{corechase}). (The traditional core chase computes finite universal models of an input instance, and it is complete.)

If $\ax$ is a cartesian theory, then the parallel chase and the categorical core chase both compute finite free models on an input instance.  The parallel chase is complete given fairness assumptions, and the core chase is complete (see Lemmas~\ref{chasecomputesinit},~\ref{chasecompleteness},~\ref{corechase}).

\subsection{The Standard and Parallel Chase}

The standard and parallel chase both have a succinct description in terms of the categorical notion of {\it pushout}.  Let us begin with a categorical description of regular logic.  Consider a regular formula $\phi(x_0,\ldots,x_n)$ on the signature $\sigma$.

\begin{definition} The {\it frozen $\phi$-instance} $\Phi$ is the instance with elements $\{x_0,\ldots,x_n\}/\sim$, where $\sim$ is the equivalence relation generated by the equational atoms of $\phi$, with sorts given by $[x_i]\in \Phi s$\footnote{$[x]$ denotes the equivalence class of $x$.} whenever $x_i:s$ in $\phi$, and the smallest relations that make $\Phi r([x_{i_0}],\ldots,[x_{i_k}])$ true for every relational atom $r(x_{i_0},\ldots,x_{i_k})$ in $\phi$.
\end{definition}


Now consider an ED $\xi$ on the signature $\sigma$:
$$
  \forall (x_0:s_0) \cdots  (x_n:s_n) \ldotp
  \phi(x_0 , \ldots , x_n) \Rightarrow
  \exists (x_{n+1}:s_{n+1}) \cdots \ (x_m:s_m)\ldotp \psi(x_0, \ldots, x_m)
$$

Let ${\sf front}$ be the frozen $\phi$-instance and ${\sf back}$ be the frozen $(\phi\wedge\psi)$-instance.  There is then a morphism $h:{\sf front}\to{\sf back}$, sending $[x_i]\mapsto [x_i]$.  This morphism is the categorical interpretation of the ED.

\begin{lemma}\label{satisfiesED}
An instance $I$ on $\sigma$ satisfies $\xi$ iff every morphism $f:{\sf front}\to I$ factors through $h$, i.e. there exists $g$ such that this diagram commutes:
\[
  \xymatrix @R=.5pc {
  {\sf front} \ar[dd]_h \ar[rd]^f \\
  & I \\
  {\sf back} \ar@{-->}[ru]^g
  }
  \]
Moreover, the instance $I$ satisfies the further constraint (\ref{cartED}) iff in this diagram $g$ exists uniquely for all $f$.
\end{lemma}

This lemma can be stated even more tersely:  $I$ satisfies $\xi$ iff the function $\sigma {\sf -Inst}(h,I)$ is surjective and (\ref{cartED}) iff $\sigma {\sf -Inst}(h,I)$ is bijective.  In categorical terminology, we call $I$ {\it weakly orthogonal} to the morphism $h$ in the former case and {\it orthogonal} to $h$ in the latter case.

\begin{definition}
  Given an instance $I$, a {\it trigger} of $\xi$ in $I$ is a morphism $f:{\sf front}\to I$.  The trigger is {\it inactive} if this morphism factors through $h$, and {\it active} if it does not.
\end{definition}

\begin{definition}
  Let $f:{\sf front}\to I$ be a trigger of $\xi$ in $I$.  Then the {\it chase step} $I\To{\mathcal{C}(\xi,f)}I'$ arises from the following pushout:
  \[
  \xymatrix {
  {\sf front} \ar[r]^f \ar[d]_h & I \ar[d]^{\mathcal{C}(\xi,f)} \\
  {\sf back} \ar[r] & I' \pullbackcorner[ul]
  }
  \]
  Now, even if the trigger $f$ was active, the trigger $\mathcal{C}(\xi,f)\circ f$ is clearly inactive, which was the ``reason'' for the chase step.
  
  Given a regular theory $\ax$ comprised of egds and tgds, a {\it standard 
  $\ax$-chase sequence} is a finite or infinite chain $I_0\to I_1\to I_2\to\cdots$ of chase steps corresponding to triggers of EDs in $\ax$.  A finite chase sequence $I_0\to\cdots\to I_n$ is {\it terminating} there is an $n$ such that there are no active triggers in $I_n$ of any ED in $\ax$.  In this case, $I_n$ is a model of $\ax$, and we call it the {\it result} of the chase.
  
  We abbreviate the composition of the path $I_k\To{\mathcal{C}(\xi_k,f_k)}\cdots\To{\mathcal{C}(\xi_{n-1},f_{n-1})}I_n$ as $I_k\To{\mathcal{C}(k,n)}I_n$.
\end{definition}

Explicitly, if $\xi$ is a tgd, then $I'$ can be constructed from $I$ by first initializing $I'$ to $I$, then adding elements ${\sf foo}_{n+1}\in I's_{n+1},\ldots,{\sf foo}_m\in I's_m$ and finally minimally extending the relations so that $I'\psi(f(x_0),\ldots,f(x_n),{\sf foo}_{n+1},\ldots,{\sf foo}_m)$ is true.  In this case, the $s$-components of $\mathcal{C}(\xi,f)$ are inclusions $Is\hookrightarrow I's$.

If $\xi$ is an egd, then $I'$ can be constructed explicitly from $I$:
\begin{itemize}
    \item For each sort $s\in S$, let $I's=Is/\sim_s$, where $\sim_s$ is the equivalence relation generated by $\{(f(x_i),f(x_j))\mid x_i=x_j\textrm{ occurs in $\psi$ with }0\leq i,j\leq n\textrm{ and }x_i,x_j:s\}$.
    \item For each relation $r\in R$ of arity $s_0,\ldots,s_n$, let $I'r$ be the image of $Ir$ under the canonical map $Is_0\times\cdots\times Ix_n\to I's_0\times\cdots\times I'x_n$.
\end{itemize}
Then the $s$-component of $\mathcal{C}(\xi,f)$ is the canonical map $Is\twoheadrightarrow I's$.

\begin{definition}
  Consider a set $\mathcal{F}$ of triggers $f:{\sf front}_f\to I$ of $\xi_f\in\ax$ in $I$, where $\xi_f$ is represented as $h_f:{\sf front}_f\to {\sf back}_f$.  Then the {\it parallel chase step} $I\To{\mathcal{C}(\mathcal{F})}I'$ arises from the following pushout, where $\bigoplus_{f\in\mathcal{F}} f$ denotes copairing:
  
  \[
    \xymatrix{
    \coprod_{f\in\mathcal{F}} {\sf front}_f \ar[r]^-{ \bigoplus_{f\in\mathcal{F}} f} \ar[d]_{\coprod_{f\in\mathcal{F}} h_f} & I \ar[d]^{\mathcal{C}(\mathcal{F})}    \\
    \coprod_{f\in\mathcal{F}} {\sf back}_f \ar[r]  & I' \pullbackcorner[ul]           
    }
  \]

  In this step, we have effectively done all of the chase steps $\mathcal{C}(\xi_f,f)$ simultaneously.  Indeed, as can be easily seen, the trigger $\mathcal{C}(\mathcal{F})\circ f$ in $I'$ is inactive for each $f\in\mathcal{F}$.
  
  In computation, we always work with finite $\mathcal{F}$, but the general case is useful for theoretical purposes.
  
  Given a regular theory $\ax$ comprised of egds and tgds, a {\it parallel
  $\ax$-chase sequence} is a finite or infinite chain $I_0\to I_1\to I_2\to\cdots$ of parallel chase steps corresponding to sets of triggers of EDs in $\ax$.  A finite parallel chase sequence $I_0\to\cdots\to I_n$ is {\it terminating} if there are no active triggers in $I_n$ of any ED in $\ax$.  In this case, $I_n$ is a model of $\ax$, and we call it the {\it result} of the chase.
  
  We again abbreviate the composition of the path $I_k\To{\mathcal{C}(\mathcal{F}_k)}\cdots\To{\mathcal{C}(\mathcal{F}_{n-1})}I_n$ as $I_k\To{\mathcal{C}(k,n)}I_n$.
\end{definition}

It is evident that the standard chase is a special case of the parallel chase, where all sets $\mathcal{F}$ have cardinality $1$.  Using this fact it is easily shown that the following results about the parallel chase apply \textit{a fortiori} to the standard chase.

It is known~\cite{Deutsch:2008:CR:1376916.1376938} that the parallel chase computes universal models.  We now refine this result slightly, showing that it computes weakly free models.

\begin{lemma}\label{chasecomputesweaklyinit}
Let $\ax$ be a regular theory (not necessarily cartesian) comprised of tgds and egds.
A terminating parallel $\ax$-chase sequence $I_0\to I_1\to \cdots\to I_n$ computes a weakly free model of $\ax$ on $I_0$.
\end{lemma}
\begin{proof}
Let $A$ be a model of $\ax$ and $a_0:I_0\to A$ be a morphism.  Given a chase step $I_0\To{\mathcal{C}(\mathcal{F}_0)}I_1$, let $f\in\mathcal{F}_0$ and consider the trigger $a_0\circ f$ of $\xi$ in $A$.  Since $A$ is a model, there is a morphism $b_f:{\sf back}_f\to A$ such that $b_f\circ h_f = a_0 \circ f$.  By the universal property of coproducts, the morphisms $b_f$ induce a morphism $b:\coprod_{f\in\mathcal{F}_0} {\sf back}_f \to A$ such that $b\circ \coprod_{f\in\mathcal{F}_0} h_f = a_0\circ [f|f\in\mathcal{F}_0]$.  Then by the universal property of pushouts, there is a unique morphism $a_1:I_1\to A$ such that the $a_1\circ \mathcal{C}(\xi_0,f_0)=a_0$ and $a_1\circ g = b$.  All this is shown in the following diagram:
\[
  \xymatrix @R=1.3pc @C=2.5pc {
  {\sf front}_f \ar[r] \ar[d]_{h_f} \ar@/^1pc/[rr]^f &
  \coprod_{f\in\mathcal{F}_0} {\sf front}_f \ar[d]_{\coprod_{f\in\mathcal{F}_0} h_f} \ar[r]_-{\bigoplus_{f\in\mathcal{F}_0} f} & I_0  \ar[d]^{\mathcal{C}(\mathcal{F}_0)} \ar@/^1pc/[rdd]^{a_0} \\
  {\sf back}_f \ar[r] \ar@{-->}@/_1pc/[rrrd]_{b_f} &
  \coprod_{f\in\mathcal{F}_0} {\sf back}_f \ar[r]^-g \ar@{-->}[rrd]^{b} & I_1 \ar@{-->}[rd]^{a_1} \pullbackcorner[ul] \\
  &&& A
  }
\]
Iterating this construction for each step of the chase sequence, we obtain a morphism $a_n:I_n\to A$ with $a_n\circ\mathcal{C}(0,n)=a_0$.  Since $A$ and $a_0$ were arbitrary, $(I_n,\mathcal{C}(0,n))$ is a weakly free model of $\ax$ on $I_0$.
\qed\end{proof}



If $\ax$ is cartesian, then the parallel chase computes free models.

\begin{proposition} \label{chasecomputesinit}
Let $\ax$ be a cartesian theory, rewritten as in Lemma \ref{cartisreg} and \ref{EDisegdsandtgds} to be comprised of egds and tgds. Then a terminating parallel $\ax$-chase sequence $I_0\to I_1\to \cdots\to I_n$ computes a free model of $\ax$ on $I_0$.
\end{proposition}
\begin{proof}

By Lemma~\ref{chasecomputesweaklyinit}, $(I_n,\mathcal{C}(0,n))$ is a weakly free model of $\ax$ on $I_0$.  For any model $B$ of $\ax$ and morphism $b:I_0\to B$, suppose that we have morphisms $p_n,q_n:I_n\to B$ with $p_n\circ\mathcal{C}(0,n)=b$ and $q_n\circ\mathcal{C}(0,n)=b$.  Let $p_k\coloneqq p_n\circ\mathcal{C}(k,n)$ and $q_k\coloneqq q_n\circ\mathcal{C}(k,n)$ for $k=0,\ldots,n$.  We prove by induction on $k$ that $p_k=q_k$ for all $k$, from which the conclusion follows.  First, $p_0=b=q_0$.  Then assume that $p_n=q_n$.  We have the following situation:

\[
\xymatrix{
\coprod_{f\in\mathcal{F}_0} {\sf front}_f \ar[d]_{\coprod_{f\in\mathcal{F}_0} h_f} \ar[r]^-{\bigoplus_{f\in\mathcal{F}_0} f} & I_n \ar@/^1pc/[dr]^{p_n=q_n} \ar[d]^{\mathcal{C}(\mathcal{F}_n)} \\
\coprod_{f\in\mathcal{F}_0} {\sf back}_f \ar[r]_-g & I_{n+1} \pullbackcorner[ul] \ar@<-.5ex>[r]_{q_{n+1}} \ar@<.5ex>[r]^{p_{n+1}} & B
}
\]

We have that $p_{n+1}\circ \mathcal{C}(\mathcal{F}_n)=p_n=q_n=q_{n+1}\circ\mathcal{C}(\mathcal{F}_n)$, so by the universal property of pushouts, it suffices to show that $p_{n+1}\circ g = q_{n+1}\circ g$.  By the universal property of coproducts, for this it suffices to show that $p_{n+1}\circ g_f=q_{n+1}\circ g_f$ for each $f\in\mathcal{F}_n$, where $g_f$ is the restriction of $g$ to ${\sf back}_f$.  We restrict the coproducts to $f$-summands for an arbitrary $f\in\mathcal{F}_n$:

\[
\xymatrix{
{\sf front}_f \ar[d]_{h_f} \ar[r]^{f} & I_n \ar@/^1pc/[dr]^{p_n=q_n} \ar[d]^{\mathcal{C}(\mathcal{F}_n)} \\
{\sf back}_f \ar[r]_{g_f} & I_{n+1} \pullbackcorner[ul] \ar@<-.5ex>[r]_{q_{n+1}} \ar@<.5ex>[r]^{p_{n+1}} & B
}
\]

We can see from this diagram that $p_n\circ f_=p_{n+1}\circ g\circ h$ and $p_n\circ f=q_{n+1}\circ g\circ h$, so since $\ax$ is cartesian, Lemma~\ref{satisfiesED} gives that $p_{n+1}\circ g = q_{n+1}\circ g$.

We also have that $p_{n+1}\circ \mathcal{C}(\xi_n,f_n) = q_{n+1}\circ \mathcal{C}(\xi_n,f_n)$, so by the universal property of pushouts, $p_{n+1}=q_{n+1}$.
\qed\end{proof}

Now we prove an infinitary form of the previous results, which we will use in our proofs of completeness (Lemmas~\ref{chasecompleteness},~\ref{corechase}).

\begin{lemma} \label{filteredfinite}
Let $I_0\To{f_0}I_1\To{f_1}\cdots$ be a sequence of morphisms of instances on a signature $\sigma$.  Let $I$ be the colimit of this sequence, with legs $l_i:I_i\to I$.  Let $a:A\to I$ be a morphism from a finite instance $A$.  Then there exists an $n$ and a morphism $a':A\to I_n$ such that $l_n\circ a' = a$.  Moreover, if we also are given a morphism $g:I_0\to A$ such that $a\circ g=l_0$, we can ensure that $a'\circ g = f_{n-1}\circ\cdots\circ f_0$ (the optionality of $g$ is denoted in the following diagram by a dotted line).
\[
\xymatrix @R=0.5pc {
& I_0 \ar[dd]^-{f_0} \ar@{.>}@/_2pc/[lddddddd]_g \ar@/^3pc/[ddddddd]^{l_0} \\
\\
& \myvdots \ar[dd]^-{f_{n-1}} \\
\\
& I_n \ar[dd]_-{f_n} \ar@/^1pc/[ddd]^{l_n} \\
\\
& \myvdots \\
A \ar[r]^a \ar@{-->}@/^1pc/[ruuu]^{a'} & I
}
\]
\end{lemma}
\begin{proof}

If the sequence is finite, terminating at $I_n$, then $I=I_n$ and the lemma follows trivially.  We proceed under the assumption that the sequence is infinite.

As in Lemma \ref{instancesascopresheafs}, we consider the instances in this sequence as functors $C\to{\sf Set}$.


The fact that colimits in functor categories are computed pointwise~\cite[Prop. 3.3.9]{riehl} gives in this case that $Ic=\colim_nI_nc$.  Recall that in ${\sf Set}$, colimits are quotients of coproducts, so in this situation $Ic=\bigsqcup_n I_nc / \sim_c$, where $\sim_c$ is an equivalence relation generated by $v\sim (f_n)_c(v)$ for $v\in I_n$.  Explicitly, $v\in I_nc$ and $v'\in I_{n'}c$ are equivalent iff there is $m\geq n,n'$ with $(f_{m-1}\circ\cdots\circ f_n)_c(v)=(f_{m-1}\circ\cdots\circ f_n')_c(v')$.  In this case, we write $v\overset{m}{\sim} v'$.

Thus for any $u\in Ac$, let $r(u)$ be a representative of the equivalence class $a(u)$ in the instance $I_{n_u}c$.  (If we were given $g$ and $u=g_c(v)$ for some $v\in I_0c$, then choose $r(u)=v$, so $n_u=0$.)  We have $(l_{n_u})_c(r(u))=a(u)$.  For any morphism $\alpha:c\to c'$ in $C$ and $u\in Ac$, we have $I_{n_u}\alpha(r(u))\overset{m_{u,\alpha}}{\sim}r(A\alpha(u))$.  Let $n$ be the maximum of $n_u$ and $m_{u,\alpha}$ over all $\alpha:c\to c'$ in $C$ and $u\in Ac$ (here we use finiteness of $A$).  Let $a'_c(u)=(f_{n-1}\circ\cdots\circ f_{n_u})_c(r(u))$.  Then $a'$ is a natural transformation and $l_n\circ a' = a$.  (If we were given $g$, then for $v\in I_0c$, $a'_c(g_c(v))=(f_{n-1}\circ\cdots\circ f_{n_{g_c(v)}})_c(v)=(f_{n-1}\circ\cdots\circ f_0)_c(v)$ by construction, so $a'\circ g = f_{n-1}\circ\cdots\circ f_0$.)

\qed \end{proof}

Experienced category theorists will recognize the above argument as a special case of the fact that filtered colimits commute with finite limits in Set~\cite[Theorem 3.8.9]{riehl}.  Indeed, by the co-Yoneda lemma~\cite{reyes_reyes_zolfaghari_2004}, the Yoneda lemma, and completeness of representables, the operation ${\sf Set}^C(A,-)$ is revealed as a finite limit:

\[
{\sf Set}^C(A,I)\cong {\sf Set}^C\left(\lim_{(c,x)\in\int A} C(c,-),I\right)
\cong \lim_{(c,x)\in\int A}{\sf Set}^C(C(c,-),I) \cong \lim_{(c,x)\in\int A} Ic
\]

\begin{definition} A parallel $\ax$-chase sequence $I_0\To{\mathcal{C}(\mathcal{F}_0)}I_1\To{\mathcal{C}(\mathcal{F}_1)}\cdots$ is {\it fair} if, for every $n$, every ED $\xi\in\ax$, and every active trigger $f:{\sf front}\to I_n$ of $\xi$, there is an $m\geq n$ such that the trigger $\mathcal{C}(n,m)\circ f$ is not active.
\end{definition}

\begin{lemma}\label{coollemma}
Let $\ax$ be a (possibly infinite) regular theory on a signature $\sigma$, rewritten as a set of egds and tgds as in Lemma \ref{EDisegdsandtgds}.  Let $I_0$ be a (possibly infinite) instance on $\sigma$.  Let $I_0\To{\mathcal{C}(\mathcal{F}_0)}I_1\To{\mathcal{C}(\mathcal{F}_1)}\cdots$ be a (possibly infinite) parallel $\ax$-chase sequence.  Let $I$ be the colimit of this sequence, with legs $l_i:I_i\to I$.  If this sequence is fair, then $(I,l_0)$ is a weakly free model of $\ax$ on $I_0$.  If additionally we have that $\ax$ is cartesian, then $(I,l_0)$ is a free model of $\ax$ on $I_0$.
\end{lemma}
\begin{proof}
For any trigger $f:{\sf front}\to I$ of $\xi\in\ax$, Lemma $\ref{filteredfinite}$ gives a morphism $f':{\sf front}\to I_n$ with $l_n\circ f' = f$.  By fairness, there is an $m\geq n$ such that the trigger $\mathcal{C}(n,m)\circ f'$ is not active.  So there is $g':{\sf back}\to I_m$ such that $g'\circ h = \mathcal{C}(n,m)\circ f'$.  Then $l_m\circ g'\circ h = l_m\circ\mathcal{C}(n,m)\circ f'=l_n\circ f' = f$, so $f$ is not active.  Thus $I$ is a model of $\ax$.

For any model $B$ and morphism $b_0:I_0\to B$, we obtain morphisms $b_n:I_n\to B$ commuting with the chase morphisms, as in Lemma \ref{chasecomputesweaklyinit}.  By the universal property of colimits, there is a unique morphism $b:I\to B$ such that $b\circ l_n = b_n$ for all $n$.  In particular, $b\circ l_0 = b_0$, so $(I,l_0)$ is weakly free.

\[
\xymatrix @R=0.5pc @C=5pc {
& I_0 \ar[dd]_-{\mathcal{C}(\mathcal{F}_0)} \ar@/^2pc/[rddddddddd]^{b_0} \\
\\
& \myvdots \ar[dd]_-{\mathcal{C}(\mathcal{F}_{n-1})} \\
\\
& I_n \ar[dd]_-{\mathcal{C}(n,m)} \ar@/^2.7pc/[ddddd]^{l_n} \ar@/^1.6pc/[rddddd]^{b_n} \\
\\
& I_m \ar[dd]_-{\mathcal{C}(\mathcal{F}_m)} \ar@/^1pc/[ddd]^{l_m} \\
{\sf front} \ar[dd]_h  \ar[rdd]^f \ar@/^1pc/[ruuu]^{f'} \\
& \myvdots \\
{\sf back} \ar@{-->}@/^1pc/[ruuu]^{g'} \ar[r]_{l_m\circ g'} & I 
\ar@{-->}[r]_b & B
}
\]

Now suppose that $\ax$ is cartesian and there are two morphisms $p,q:I\to B$ with $p\circ l_0=b_0$ and $q\circ l_0=b_0$.  We intend to show that $p=q$.  By the universal property of colimits, it suffices to show that $p_n\coloneqq p\circ l_n = q\circ l_n \eqqcolon q_n$ for all $n$, which can be shown by induction just as in Proposition~\ref{chasecomputesinit}.
\qed\end{proof}

\begin{corollary}\label{existence}
Let $\ax$ be a (possibly infinite) regular theory on a (possibly infinite) signature $\sigma$ and let $I_0$ be a (possibly infinite) instance on $\sigma$.  Then there exists a (possibly infinite) weakly free model of $\ax$ on $I_0$.  If $\ax$ is cartesian, then there exists a (possibly infinite) free model of $\ax$ on $I_0$.
\end{corollary}

\begin{proof}
It suffices to show that there exists a fair parallel $\ax$-chase sequence $I_0\To{\mathcal{C}(\mathcal{F}_0)}I_1\To{\mathcal{C}(\mathcal{F}_1)}\cdots$.  Simply let $\mathcal{F}_n$ be the entire set of triggers in $I_n$ of EDs in $\ax$, for each $n$.
\qed\end{proof}

Thus, by Lemma~\ref{lemma:adjunction}, if $\ax$ is a cartesian theory, ${\sf Mod}(\ax)$ is a reflective subcategory of $\sigma{\sf -Inst}$.  In the general case of a regular theory $\ax$, we say that ${\sf Mod}(\ax)$ is a {\it weakly reflective subcategory}~\cite{weakly_reflective} of $\sigma{\sf -Inst}$, and the inclusion functor is said to have a {\it weak left adjoint}~\cite{kainen_1971}.

The following proposition generalizes Theorem 6.1 in \cite{BUSH2003107} (see Section~\ref{prev}).

\begin{proposition}\label{chasecompleteness}
Let $\ax$ be a cartesian theory, rewritten as in Lemma \ref{cartisreg} and \ref{EDisegdsandtgds} to be comprised of egds and tgds, and let $I_0\to I_1\to \cdots$ be a parallel $\ax$-chase sequence such that
\begin{enumerate}
    \item Each $\mathcal{F}_n$ contains at least one active trigger.
    \item The sequence $I_0\to I_1\to\cdots$ is fair.
    \item For every $n$ there is an $m\geq n$ such that $I_m$ satisfies all egds in $\ax$.
    \item There is a finite weakly free model of $\ax$ on $I_0$.
\end{enumerate}

Then the chase sequence $I_0\to I_1\to \cdots$ terminates.
\end{proposition}
\begin{proof}
Let $I$ be the colimit of the chase sequence, with legs $l_i:I_i\to I$.  Since the sequence is fair (2nd assumption) and $\ax$ is cartesian, Lemma~\ref{coollemma} gives that $(I,l_0)$ is a free model of $\ax$ on $I_0$.  By the 4th assumption and Lemma~\ref{finiteweaklyinitial}, $I$ must be finite.

By Lemma~\ref{filteredfinite} applied to the identity morphism ${\sf id}_I$ and the leg $l_0$, there is an $n$ and a morphism $i_n:I\to I_n$ with $l_n\circ i_n={\sf id}_I$ and $i_n\circ l_0=\mathcal{C}(0,n)$.  Now use the 3rd assumption to find an $m\geq n$ such that $I_m$ satisfies all egds in $\ax$.  Let $i_m=\mathcal{C}(n,m)\circ i_n$, and note that $l_m\circ i_m={\sf id}_I$ and $i_m\circ l_0=\mathcal{C}(0,m)$.  We now show that $i_m$ is an isomorphism, which implies that $(I_m,\mathcal{C}(0,m))$ is a free model of $\ax$ on $I_0$, which by the 1st assumption implies that the chase sequence terminates at $I_m$. \newpage 

\vspace*{-.3in}
\[
\xymatrix @R=0.5pc @C=5pc {
& I_0 \ar[dd]^-{\mathcal{C}(\mathcal{F}_0)} \ar@/_2pc/[lddddddddd]_{l_0} 
\\
\\
& \myvdots \ar[dd]^-{\mathcal{C}(\mathcal{F}_{n-1})} \\
\\
& I_n \ar[dd]_-{\mathcal{C}(n,m)} \ar@/^3pc/[ddddd]^{l_n} \\
\\
& I_m \ar[dd] \ar@/^1pc/[ddd]^{l_m} \\
\\
& \myvdots \\
I \ar@/^1pc/[ruuu]^{i_m} \ar@/^1.5pc/[ruuuuu]^{i_n} \ar@{=}[r] & I
}
\]
\vspace*{-.05in}

Clearly $i_m$ is injective.  If $I_mr(i_m(v_1),\ldots,i_m(v_k))$ holds for $x_1,\ldots,v_k\in I$, then $l_mi_m(v_j)=v_j$ for each $j=1,\ldots,k$ so $Ir(v_1,\ldots,v_k)$ holds.  So to show that $i_m$ is an isomorphism it suffices to show that each of its components is surjective.

We now define the {\it rank} of an element $v\in I_ms$ as the smallest $k$ such that $v$ is in the image of $\mathcal{C}(k,m)_s$.

If not all components of $i_m$ are surjective, let $v$ be an element of $I_m$, not in the image of $i_m$, of minimal rank $k$, and let $u$ be an element of $I_k$ with $\mathcal{C}(k,m)(u)=v$ (we omit the subscript $s$ hereafter for convenience). If $k=0$, then $v=\mathcal{C}(0,m)(u)=i_ml_0(u)$, a contradiction.  So $k>0$.  It is then clear that $u$ was introduced from $\coprod_{f\in\mathcal{F}_{k-1}}{\sf back}_f$ in the following pushout: \vspace*{-.15in}
  \[
    \xymatrix @C=3.5pc {
    \coprod_{f\in\mathcal{F}_{k-1}} {\sf front}_f \ar[r]^-{\bigoplus_{f\in\mathcal{F}_{k-1}} f} \ar[d]_{\coprod_{f\in\mathcal{F}_{k-1}} h_f} & I_{k-1} \ar[d]^{\mathcal{C}(\mathcal{F}_{n-1})}    \\
    \coprod_{f\in\mathcal{F}_{n-1}} {\sf back}_f \ar[r]_-g  & I_k \pullbackcorner[ul]           
    }
  \]
Thus there is a trigger $f\in\mathcal{F}_{k-1}:{\sf front}_f\to I_{k-1}$ of an tgd $\xi$:
\begin{small}
\[
\phi(x_0 , \ldots , x_\nu) \Rightarrow
  \exists (x_{\nu+1}:s_{\nu+1}) \cdots \ (x_\mu:s_\mu)\ldotp \psi(x_0, \ldots, x_\mu)
\]
\end{small}
and $u=g(x_p)$ for some $p=\nu+1,\ldots,\mu$.

From now on, we write $i_m$ as $i$ and $l_m$ as $l$ to simplify notation.  By minimality of $k$, every element of the image of $\mathcal{C}(k-1,m)_s\circ f_s$ is in the image of $i_s$, for each sort $s$.  Define functions $f'_s:{\sf front}_fs\to Is$ by $i_s(f'_s(x_j)) = \mathcal{C}(k-1,n)_s (f_s(x_j))$ for $j=0,\ldots,\nu$. Since $I_m\phi(i(f'(x_0)),\ldots,i(f'(x_\nu)))$ holds and we have $l\circ i={\sf id}_A$, we have $I\phi(f'(x_0),\ldots,f'(x_\nu))$, so $f'$ is a morphism.  Since $I$ is a model of $\ax$, there is a morphism $g':{\sf back}\to A$ with $g'\circ h=f'$.  All this is shown in the following diagram, where $g_f$ is the restriction of $g$ to ${\sf back}_f$:
\vspace{-.15in}
\[
  \xymatrix {
  {\sf front}_f \ar@/_1.5pc/[dd]_{f'} \ar[r]^f \ar[d]^{h_f} & I_{k-1} \ar[d]^{\mathcal{C}(\mathcal{F}_{k-1})} \\
  {\sf back}_f \ar@{-->}[d]^{g'} \ar[r]_{g_f} & I_k \ar[d]^{\mathcal{C}(k,m)} \\
  I \ar@{^{(}->}[r]^{i} & I_m
  }
\]
Thus we have both $I_m\psi(if'(x_0),\ldots,if'(x_\nu),ig'(x_{\nu+1}),\ldots,ig'(x_\mu))$ and\linebreak $I_m\psi(if'(x_0),\ldots,if'(x_\nu),\mathcal{C}(k,m)g(x_{\nu+1}),\ldots,\mathcal{C}(k,m)g(x_\mu))$.  Since $\ax$ is cartesian, it has an egd expressing the uniqueness of $\xi$; and $I_m$ satisfies all egds in $\ax$, so it must satisfy this one.  Thus $ig'(x_p)=\mathcal{C}(k,m)g(x_p)=v$, contradicting $v\notin i_s(Is)$.  Thus each component of $i$ is surjective, and the chase sequence terminates at $I_m$.

\qed \end{proof}

So although the standard chase is incomplete in general, it is complete when $\ax$ is cartesian and we only care about computing finite free models.  We still do not have completeness with respect to computing finite universal models, even when $\ax$ is cartesian (see Example~\ref{universalcounterexample}).  We also do not have completeness with respect to computing finite weakly free models of regular theories --- Example 2 in \cite{Deutsch:2008:CR:1376916.1376938} has a finite weakly free model and no terminating chase sequence.  For completeness, we need both $\ax$ to be cartesian and our goal to be finite weakly free models.

\subsection{A Fast Parallel Chase Algorithm for Cartesian Theories}
\label{canchase}
We now describe a particular canonical (determined up to isomorphism) chase algorithm for cartesian theories satisfying criteria 1, 2, and 3 of Lemma~\ref{chasecompleteness}.

Let $\ax$ be a finite cartesian theory on the signature $\sigma$, rewritten as in Lemma \ref{cartisreg} and \ref{EDisegdsandtgds} to be comprised of egds and tgds, and let $I_0$ be a finite instance on $\sigma$.

\begin{algorithm}[H]
\setstretch{1.35}
\SetAlgoLined
\caption{Fast Parallel Chase Algorithm}\label{alg:fast}
\KwData{$\sigma$, $\ax$, $I_0$}
\KwResult{$I$, $i:I_0\to I$}
$I_{\textrm{cur}}\coloneqq I_0$\;
$i\coloneqq {\sf id}_{I_0}$\;
${\sf first}\coloneqq {\sf true}$\;
\While{$I_{\textrm{cur}}$ has active triggers of tgds OR ${\sf first}$}{
  ${\sf first}\coloneqq {\sf false}$\;
  parallel chase $I_{\textrm{cur}}\To{\mathcal{C}(\textrm{all active triggers of tgds in }\ax)}I$\;
  $i\coloneqq\textrm{compose}(I_0\To{i}I_{\textrm{cur}}\to I)$\;
  $I_{\textrm{cur}}\coloneqq I$\;
  \While{$I_{\textrm{cur}}$ has active triggers of egds}{
    parallel chase $I_{\textrm{cur}}\To{\mathcal{C}(\textrm{all active triggers of egds in }\ax)}I$\;
    $i\coloneqq\textrm{compose}(I_0\To{i}I_{\textrm{cur}}\to I)$\;
    $I_{\textrm{cur}}\coloneqq I$\;
  }
}
\end{algorithm}

This algorithm terminates whenever the free model of $\ax$ on $I$ is finite, and in such a case it computes this model.

If we just want to compute the model $I$ and don't care about the morphism $i$, as in the case of left Kan extensions (see Lemma~\ref{quadruple}) then we can omit steps 2, 5, and 9.

\pagebreak 

\section{A Fast Left Kan Extension Algorithm}
\label{impl}

We have implemented a specialized version of the parallel chase algorithm in Section~\ref{canchase} tailored for computing left Kan extensions inside the open-source CQL tool for computational category theory ({\sf http://categoricaldata.net}); this algorithm also resembles a parallel version of the left Kan algorithm in~\cite{BUSH2003107}.  By leveraging collection-oriented (parallelizable) operations, it improves upon our optimized (necessarily sequential) java implementations of various existing left Kan algorithms (\cite{CARMODY1995459} and \cite{BUSH2003107} and several from \cite{patrick}) by a factor of ten on our benchmarks.  




\subsection{Input Specification}

The input to our left Kan algorithm consists of a source finite category presentation, whose vertex (node) set we refer to as $C$, whose edge set from $c_1$ to $c_2$ we refer to as $C(c_1, c_2)$, and whose equations (pairs of possibly 0-length paths) from $c_1$ to $c_2$ we refer to as $CE(c_1, c_2)$.  Similarly, our input contains a target finite category presentation ($D$, $D(-,-)$,  $DE(-,-)$).   We require as further input a morphism of presentations $F : (C, C(-,-), CE(-,-)) \to (D, D(-,-), DE(-,-) $.  Finally, we require a $(C,C(-,-))$-algebra $I$ satisfying $CE$.  The source equations $CE$ are not used by our algorithm (or any chase algorithm we are aware of) but are required to fully specify $I$.  A review of concepts such as finite category presentation and $(C,C(-,-))$-algebra is found in Section~\ref{section.ct}.



\subsection{The State}

Like most chase algorithms~\cite{onet:DFU:2013:4288}, our left Kan extension algorithm runs in rounds, possibly forever, transforming a state 
consisting of an $\sigma$-instance (where $\sigma$ is the signature of the theory ${\sf cog}(F)$)
until a fixed point is reached.  In general, termination of the chase is undecidable, but sufficient criteria exist based on the acyclicity of the ``firing pattern'' of the existential quantifiers~\cite{onet:DFU:2013:4288} in the cartesian theory corresponding to $DE$ from the previous section.  Similarly, termination of a left Kan extension is undecidable~\cite{CARMODY1995459}, although using the results of this paper we can use termination of one to show termination of the other.  Formally, the state of our algorithm consists of:

\begin{itemize}
\item For each $d \in D$, a set $J(d)$, the elements of which we call {\it output rows}.  $J$ is initialized in the first round by setting $J(d) \coloneqq \bigsqcup_{\{c \in C \ | \ F(c) = d\}} I(c)$.
\item For each $d \in D$, an equivalence relation $\sim_d \ \subseteq  \ J(d) \times J(d)$, initialized to identity at the beginning of every round.
\item For each edge $f : d_1 \to d_2 \in D$, a binary relation $J(f) \subseteq J(d_1) \times J(d_2)$, initialized in the first round to empty.  When the chase completes, each such relation will be total and deterministic.
\item For each node $c \in C$, a function $\eta(c) : I(c) \to J(F(c))$, initialized in the first round to the coproduct/disjoint-union injections from the first item, i.e. $\eta(c)(v)=(c,v)$.
\end{itemize}


Given a morphism $f:d_1\to d_2$, we may {\it evaluate} $p$ on any $u\in J(d_1)$, written $p(u)$, resulting in a (possibly empty) set of values $\{v\in J(d_2)\mid (u,v)\in J(f)\}$.

We may also evaluate a path $p=(d_0\To{f_1}d_2\To{f_2}\cdots\To{f_n}d_n)$ on $u$.  We can define this inductively as $p(u)=\{u\}$ when $n=0$ and $p(u)=\bigcup_{v\in p'(u)}f_n(v)$ where $p'$ is the subpath $d_0\To{f_1}d_2\To{f_2}\cdots\To{f_{n_1}}d_{n_1})$.


\subsection{The Algorithm}

Unlike most chase algorithms, our left Kan algorithm consists of a fully deterministic sequence of state transformations, up to unique isomorphism. In practice, designing a practical chase algorithm (and hence, we believe, a practical left Kan algorithm) comes down to choosing an equivalent sequence of state transformations, as well as efficiently executing them in bulk.  
  In this section, we first describe the actions of our algorithm, and then we show how each action is equivalent to a sequence of the two kinds of actions used to define chase steps in database theory.  We conclude with some discussion around the additional observation that each step in our left Kan algorithm can be directly read as computing a pushout in the chase algorithm of Section~\ref{canchase}.

A single step of our left Kan algorithm involves applying the following actions (named after the actions in \cite{BUSH2003107}) to the state in the order dictated by Algorithm \ref{alg:fastkan}.
   
\begin{itemize}
\item Action $\alpha$: {\it add new elements}.   For every edge $g : d_1 \to d_2$ in $D$ and $u \in J(d_1)$ for which there does not exist $v \in J(d_2)$ with $(u,v) \in J(g)$, add a fresh (not occurring elsewhere) symbol ${\sf g(u)}$ to $J(d_2)$, and add $(u,{\sf g(u)})$ to $J(g)$. Note that this action may not make every relation $J(g)$ total (which might necessitate an infinite chain of new element creations), but rather adds one more ``layer'' of new elements.

\item Action $\beta_D$: {\it add all coincidences induced by $D$}.  In this step, for each equation $p = q$ in $DE(d_1,d_2)$ and $u \in J(d_1)$, we update $\sim_{d_2}$ to be the smallest equivalence relation also including $\{(v,v')\mid v\in p(u),v'\in q(u)\}$.
\item Action $\beta_F$: {\it add all coincidences induced by $F$}.  For each edge $f:c\to c'$ in $C$ and $u\in I(c)$, update $\sim_{Fc'}$ to be the smallest equivalence relation also including $\{((c',f(u)),v)\ | \ v \in f((c,u))\}$.
\item Action $\delta$: {\it add all coincidences induced by functionality}.  For each $f:d_1\to d_2$ and every $(u,v)$ and $(u,v')$ in $J(f)$ with $v \neq v'$, update $\sim_{d_2}$ to be the smallest equivalence relation also including $(v,v')$.  
\item Action $\gamma$: {\it merge coincidentally equal elements}.  At the end of every round, we replace every entry in $J$ and $\eta$ with its $\sim$-equivalence class (or $\sim$-representative).
\end{itemize}

The phrase ``add coincidences'' is used by the authors of~\cite{BUSH2003107} where a database theorist would use the phrase ``fire equality-generating dependencies''.  

In many chase algorithms, including~\cite{BUSH2003107}, elements are equated in place, necessitating complex reasoning and inducing non-determinism.  Our algorithm is deterministic: action $\alpha$ adds a new layer of elements, and the next steps add to $\sim$.  

\vspace{.25in}
\SetKwComment{Comment}{// }{}
\begin{algorithm}[H]
\setstretch{1.1}
\SetAlgoLined
\caption{Fast Left Kan Algorithm}\label{alg:fastkan}
\KwData{$C$, $D$, $F$, $I$, $\eta$}
\KwResult{$J$, $\eta$}
\For{$c\in C$ and $v\in Ic$}{
  add $(c,v)$ to $J(Fc)$\;
  $\eta(c)(v)\coloneqq (c,v)$\;
}
${\sf first}\coloneqq {\sf true}$\;
\While{$J$ has active triggers of tgds OR ${\sf first}$}{
  ${\sf first}\coloneqq {\sf false}$\;
  \For{$d\in D$}{
    $\sim_d\coloneqq\{(v,v)\mid v\in J(d)\}$\;
  }
  $\alpha$\;
  \While{$J$ has active triggers of egds}{
    $\beta_D$\;
    $\beta_F$\;
    $\delta$\;
    $\gamma$\;
  }
}

\end{algorithm}

\begin{proposition}
Algorithm~\ref{alg:fastkan} instantiates Algorithm~\ref{alg:fast} with the theory ${\sf cog}(F)$ and the input instance $\mathcal{I}$ defined as in Section~\ref{prf}, so it computes finite left Kan extensions completely.
\end{proposition}
\begin{proof}
By construction; we designed actions $\alpha, \beta_D, \beta_F, \delta, \gamma$ by grouping the EDs for ${\sf cog}(F)$ according to where they came from.  Algorithm~\ref{alg:fastkan} performs all the tgds at once in line 8, then prepares all the egds in lines 11-13, performing them simultaneously in line 14.
\qed\end{proof}

 \subsection{Example Run of the Algorithm}
See section~\ref{sec:ex} for the definition of our running example.  The state begins as:
\[
\begin{tabular}{ll}
 {\sf Faculty} & {\sf isFP}     \\\hline 
  {\sf  'Dr.' Alice } &  \\ 
  {\sf  'Dr.' Bob } &   \\ 
  {\sf  Prof. Ed } &  \\ 
  {\sf  Prof. Finn } &  \\ 
  {\sf  Prof. Gil } &  \\
\end{tabular}
\hspace{.25in}
\begin{tabular}{ll}
 {\sf Student} & {\sf isSP}    \\\hline 
 {\sf   Alice } & \\ 
{\sf   Bob } &  \\ 
{\sf   Chad } & \\
{\sf   Doug } &  \\ \\
\end{tabular}
\hspace{.25in}
\begin{tabular}{lll}
  {\sf TA}  &  {\sf isTF}  &  {\sf isTS}  \\\hline 
{\sf  math-TA }&{\sf  }&{\sf  }\\ 
{\sf  cs-TA  }&{\sf  }&{\sf  } \\ \\ \\ \\
\end{tabular}
\hspace{.25in}
\begin{tabular}{ll}
  {\sf Person}     \\\hline 
  \\ \\ \\ \\ \\
\end{tabular}
\]
The $\eta$ tables are initialized to identities and do not change, so we do not display them.  First, we add new elements (action $\alpha$):
\[
\begin{tabular}{ll}
 {\sf Faculty} & {\sf isFP}     \\\hline 
  {\sf  'Dr.' Alice } & {\sf isFP}({\sf 'Dr.' Alice})  \\ 
  {\sf  'Dr.' Bob } &  {\sf isFP}({\sf 'Dr.' Bob}) \\ 
  {\sf  Prof. Ed } & {\sf isFP}({\sf  Prof. Ed}) \\ 
  {\sf  Prof. Finn } & {\sf isFP}({\sf  Prof. Finn}) \\ 
  {\sf  Prof. Gil } & {\sf isFP}({\sf Prof. Gil}) \\
  {\sf isTF}({\sf math-TA}) & \\
  {\sf isTF}({\sf cs-TA }) & \\
\end{tabular}
\hspace{.5in}
\begin{tabular}{ll}
 {\sf Student} & {\sf isSP}    \\\hline 
 {\sf   Alice } & {\sf isSP}({\sf Alice}) \\ 
{\sf   Bob } & {\sf isSP}({\sf Bob}) \\ 
{\sf   Chad } &  {\sf isSP}({\sf Chad}) \\ 
{\sf   Doug } & {\sf isSP}({\sf Doug}) \\ 
{\sf isTS}({\sf math-TA}) & \\
{\sf isTS}({\sf cs-TA }) & \\
\\
\end{tabular}
\]
\[
\begin{tabular}{lll}
  {\sf TA}  &  {\sf isTF}  &  {\sf isTS}  \\\hline 
{\sf  math-TA }& {\sf isTF}({\sf math-TA})&{\sf isTS}({\sf math-TA})\\ 
{\sf  cs-TA  }&{\sf isTF}({\sf cs-TA })&{\sf isTS}({\sf cs-TA })
\end{tabular}
\hspace{.5in}
\begin{tabular}{ll}
  {\sf Person}     \\\hline 
   {\sf isFP}({\sf 'Dr.' Alice})  \\ 
   {\sf isFP}({\sf 'Dr.' Bob}) \\ 
   {\sf isFP}({\sf  Prof. Ed}) \\ 
   {\sf isFP}({\sf  Prof. Finn}) \\ 
   {\sf isFP}({\sf Prof. Gil}) \\
  {\sf isSP}({\sf Chad}) \\ 
 {\sf isSP}({\sf Alice}) \\ 
 {\sf isSP}({\sf Bob}) \\ 
 {\sf isSP}({\sf Doug}) \\ 
\end{tabular}
\]
Next, we add coincidences (actions $\beta_D$,$\beta_F$,and $\delta$).  The single target equation in $DE$ induces no equivalences, because of the missing values in the {\sf isFP} and {\sf isSP} columns, so $\beta_D$ does not apply.  $\beta_F$ requires that {\sf isTF} and {\sf isTS} be copies of {\sf isTF'} and {\sf isTS'} (from the source schema $C$), inducing the following equivalences:
\begin{gather*}
{\sf isTF}(\text{\sf math-TA}) \sim \text{\sf 'Dr.' Alice} \ \ \ \ {\sf isTS}(\text{\sf math-TA}) \sim \text{\sf Alice} 
\\
{\sf isTF}(\text{\sf cs-TA }) \sim  \text{\sf 'Dr.' Bob}\ \ \ \ {\sf isTS}(\text{\sf cs-TA }) \sim \text{\sf Bob} 
\end{gather*}
The edge relations are all functions, so action $\delta$ does not apply.  So, after merging equal elements (action $\gamma$) we have:
\[
\begin{tabular}{ll}
 {\sf Faculty} & {\sf isFP}     \\\hline 
  {\sf  'Dr.' Alice } & {\sf isFP}({\sf 'Dr.' Alice})  \\ 
  {\sf  'Dr.' Bob } &  {\sf isFP}({\sf 'Dr.' Bob}) \\ 
  {\sf  Prof. Ed } & {\sf isFP}({\sf  Prof. Ed}) \\ 
  {\sf  Prof. Finn } & {\sf isFP}({\sf  Prof. Finn}) \\ 
  {\sf  Prof. Gil } & {\sf isFP}({\sf Prof. Gil}) \\
\end{tabular}
\hspace{.5in}
\begin{tabular}{ll}
 {\sf Student} & {\sf isSP}    \\\hline 
 {\sf   Alice } & {\sf isSP}({\sf Alice}) \\ 
{\sf   Bob } & {\sf isSP}({\sf Bob}) \\ 
{\sf   Chad } &  {\sf isSP}({\sf Chad}) \\ 
{\sf   Doug } & {\sf isSP}({\sf Doug}) \\ 
\end{tabular}
\]
\[
\begin{tabular}{lll}
  {\sf TA}  &  {\sf isTF}  &  {\sf isTS}  \\\hline 
{\sf  math-TA }& {\sf 'Dr.' Alice}&{\sf Alice}\\ 
{\sf  cs-TA  }&{\sf 'Dr.' Bob}&{\sf Bob}
\end{tabular}
\hspace{.5in}
\begin{tabular}{ll}
  {\sf Person}     \\\hline 
   {\sf isFP}({\sf 'Dr.' Alice})  \\ 
   {\sf isFP}({\sf 'Dr.' Bob}) \\ 
   {\sf isFP}({\sf  Prof. Ed}) \\ 
   {\sf isFP}({\sf  Prof. Finn}) \\ 
   {\sf isFP}({\sf Prof. Gil}) \\
  {\sf isSP}({\sf Chad}) \\ 
 {\sf isSP}({\sf Alice}) \\ 
 {\sf isSP}({\sf Bob}) \\ 
 {\sf isSP}({\sf Doug}) \\ 
\end{tabular}
\]
The reason that the empty cells disappear in the {\sf Faculty} and {\sf Student} tables is that after merging, they are subsumed by existing entries; i.e., based on our definition of state, the tables below are visually distinct but completely the same as the tables above:
\[
\begin{tabular}{ll}
 {\sf Faculty} & {\sf isFP}     \\\hline 
  {\sf  'Dr.' Alice } & {\sf isFP}({\sf 'Dr.' Alice})  \\ 
  {\sf  'Dr.' Bob } &  {\sf isFP}({\sf 'Dr.' Bob}) \\ 
  {\sf  Prof. Ed } & {\sf isFP}({\sf  Prof. Ed}) \\ 
  {\sf  Prof. Finn } & {\sf isFP}({\sf  Prof. Finn}) \\ 
  {\sf  Prof. Gil } & {\sf isFP}({\sf Prof. Gil}) \\
  {\sf 'Dr.' Alice} & \\
  {\sf 'Dr.' Bob} & \\
\end{tabular}
\hspace{.5in}
\begin{tabular}{ll}
 {\sf Student} & {\sf isSP}    \\\hline 
 {\sf   Alice } & {\sf isSP}({\sf Alice}) \\ 
{\sf   Bob } & {\sf isSP}({\sf Bob}) \\ 
{\sf   Chad } &  {\sf isSP}({\sf Chad}) \\ 
{\sf   Doug } & {\sf isSP}({\sf Doug}) \\ 
{\sf Alice} & \\
{\sf Bob} & \\
\end{tabular}
\]
In the second and final round, no new elements are added and one action adds coincidences, $\beta_D$.  In particular, it induces equivalences
$$
{\sf isFP}(\text{\sf 'Dr.' Alice}) \sim {\sf isSP}(\text{\sf Alice}) \ \ \ \  \ \  \ \
   {\sf isFP}(\text{\sf 'Dr.' Bob}) \sim {\sf isSP}(\text{\sf Bob})
$$
which, after merging, leads to a final state of:
\[
\begin{tabular}{ll}
 {\sf Faculty} & {\sf isFP}     \\\hline 
  {\sf  'Dr.' Alice } & {\sf isFP}({\sf 'Dr.' Alice})  \\ 
  {\sf  'Dr.' Bob } &  {\sf isFP}({\sf 'Dr.' Bob}) \\ 
  {\sf  Prof. Ed } & {\sf isFP}({\sf  Prof. Ed}) \\ 
  {\sf  Prof. Finn } & {\sf isFP}({\sf  Prof. Finn}) \\ 
  {\sf  Prof. Gil } & {\sf isFP}({\sf Prof. Gil}) \\
\end{tabular}
\hspace{.5in}
\begin{tabular}{ll}
 {\sf Student} & {\sf isSP}    \\\hline 
 {\sf   Alice } & {\sf isSP}({\sf Alice}) \\ 
{\sf   Bob } & {\sf isSP}({\sf Bob}) \\ 
{\sf   Chad } &  {\sf isSP}({\sf Chad}) \\ 
{\sf   Doug } & {\sf isSP}({\sf Doug}) \\ 
\end{tabular}
\]
\[
\begin{tabular}{lll}
  {\sf TA}  &  {\sf isTF}  &  {\sf isTS}  \\\hline 
{\sf  math-TA }& {\sf 'Dr.' Alice}&{\sf Alice}\\ 
{\sf  cs-TA  }&{\sf 'Dr.' Bob}&{\sf Bob}
\end{tabular}
\hspace{.5in}
\begin{tabular}{ll}
  {\sf Person}     \\\hline 
   {\sf isFP}({\sf 'Dr.' Alice})  \\ 
   {\sf isFP}({\sf 'Dr.' Bob}) \\ 
   {\sf isFP}({\sf  Prof. Ed}) \\ 
   {\sf isFP}({\sf  Prof. Finn}) \\ 
   {\sf isFP}({\sf Prof. Gil}) \\
  {\sf isSP}({\sf Chad}) \\ 
 {\sf isSP}({\sf Doug}) \\ 
\end{tabular}
\]
which is obviously uniquely isomorphic\footnote{Note that the uniqueness of the isomorphism from these tables to any other left Kan extension does not imply there are no automorphisms of the input (for example, we may certainly swap {\sf Prof. Finn} and {\sf Prof. Gil}), but rather that there will be at most one isomorphism of the tables above with any other left Kan extension.  If desired, we may of course restrict ourselves to only considering those morphisms that leave their inputs fixed, ruling out the swapping of {\sf Prof. Finn} and {\sf Prof. Gil}, described in various places in this paper as ``leaving {\sf Prof. Finn}, {\sf Prof. Gill} as constants''.} to the original example output (see Section~\ref{sec:ex}).  The actual choice of names in the above tables is not canonical, as we would expect for a set-valued functor defined by a universal property, and different naming strategies are possible.  


\subsection{Comparison to Previous Work} 
\label{prev}

The authors of~\cite{BUSH2003107} identify four actions that leave invariant the left Kan extension denoted by a state, and consider sequences of these actions.  We compare their actions with ours:  
\begin{itemize}
\item Action $\alpha$: add a new element.  This step is similar to our $\alpha$ step, except it only adds one element. 
\item Action $\beta$: add a coincidence. This step is similar to our $\beta_F$ and $\beta_D$, except it only considers one equation.
\item Action $\delta$: delete non-determinism.  This is similar to our $\delta$ step, except it only applies to one edge at a time.  If $(u,v) \in J(g)$ and $(u,v') \in J(g)$ but $v \neq v'$, add $(v,v')$ and $(v',v)$ to $\sim$ and delete $(u,v')$ from $J(g)$.  This process is biased towards keeping older values to ensure fairness.  
\item Action $\gamma$: delete a coincidence.  If $(u, v) \in \ \sim_d$ for some $d \in D$, then replace $v$ by $u$ in various places, and add new coincidences.  In the first computational left Kan paper~\cite{CARMODY1995459}, this action took an entire companion technical report to justify~\cite{10.1007/BFb0084213}; the authors of~\cite{BUSH2003107} reduced this step to about a page.  One reason this step is complicated to write in~\cite{BUSH2003107} is because the relation $\sim$ is not required to be transitive; another reason is that the way deletion is done in the various places depends on the particular place; another is that deletion is done in place.  Finally, in~\cite{BUSH2003107}, $\sim$ is persistent; their notion of action and round are the same, and $\sim$ does not reset between rounds.  
\end{itemize}

Theorem 6.1 in \cite{BUSH2003107} (paraphrased) states that a sequence of four actions $\alpha$, $\beta$, $\gamma$, and $\delta$ terminates if
\begin{enumerate}
\item For each action $\eta=\alpha,\beta,\gamma$, or $\delta$ and each $n\geq 1$, there exists $m$ such that $m>n$ and $\eta_m=\eta$ (i.e. no action is left out of the sequence indefinitely).
\item When applying action $\alpha$ the element involved is always chosen to have minimal rank.
\item When applying action $\gamma$ the element of highest rank in the coincidence is deleted.
\item For all $n\geq 1$ there exists $m$ such that $m>n$ and the $m$th state satisfies all egds.
\item There is a finite left Kan extension.
\end{enumerate}

We can see that our Proposition~\ref{chasecompleteness} generalizes this by (1) generalizing condition 1 to ``fairness'', (2) discarding conditions 2 and 3 concerning rank, and (3) applying to all cartesian theories, not just ones representing left Kan extensions.

In our work on a categorical treatment of the chase, we have found that it bears a resemblance to the ``small object argument'' of category theory (see~\cite{garner}).  The problem discussed in \cite[Chapter 6]{kelly_enriched} of the reflectivity of categories of models of essentially algebraic theories resembles the problem of the existence of free models of theories on input instances.  The Ehresmann-Kennison theorem~\cite[Chapter 4]{ttt} is also a similar result, proved on finite-limit sketches, which implies our result by an argument in the Appendix.  More work is needed to fully integrate these results into database theory.

\subsection{Implementation and Experiments in CQL} 

In this section we establish the baseline performance of our algorithm by a reference implementation, primarily motivated by the fact that we are not aware of {\it any} benchmarks for {\it any} left Kan algorithms besides our own from previous work~\cite{patrick}.  

The primary optimization of our CQL implementation of our left Kan algorithm is to minimize memory usage by storing cardinalities and lists instead of sets, such that a CQL left Kan state as benchmarked in this paper consists of:
\begin{enumerate}
\item For each $d \in D$, a number $J(d) \geq 0$ representing the cardinality of a set. 
\item For each $d \in D$, a union-find data structure~\cite{Nelson:1980:FDP:322186.322198} based on path-compressed trees $\sim_d \ \subseteq \{ n \ | \ 0 \leq n < J(d) \} \times \{ n \ | \ 0 \leq n < J(d) \}$~\cite{DBLP:books/daglib/0037819}.
\item For each edge $f : d_1 \to d_2 \in D$, a list of length $J(d_1)$, each element of which is a set of numbers $\geq 0$ and $< J(d_2)$.
\item For each $c \in C$, a function $\eta(c) : I(c) \to \{ n \ | \ 0 \leq n < J(F(c)) \}$.
\end{enumerate}
From a theoretical viewpoint, the above state is more precisely considered as modeling a functor to the {\it skeleton}~\cite{BW} of the category of sets.  



Scalability tests, for both time (rows/second) and space (rows/megabyte(MB) of RAM) based on randomly constructed instances of the running example taken on a 13'' 2018 MacBook Air with a 1.6GHZ i5 CPU and 16GB RAM, on Oracle Java 11, are shown in Figure~\ref{fig4}.  Perhaps not as familiar as time throughput, memory throughput, measured here in rows/MB, measures the memory used by the algorithm during its execution as a function of input size; the periodic spikes in Figure~\ref{fig4} are likely due to the ``double when size exceeded'' behavior of the many hash-set and hash-map data structures~\cite{DBLP:books/daglib/0037819} in our Java implementation.  Memory throughput improves as the input gets larger, we believe, because the path-compressed union-find data structure of item two above scales logarithmically in space.  Time throughput (rows / second) gets worse as the input gets larger, we believe, because that same union-find structure scales linearly times logarithmically in time.  The CQL implementation runs the Java garbage collector between rounds, uses ``hash-consed''~\cite{Baader:1998:TR:280474}, tree-based terms, and uses strings for symbol and variable names.  Although performance on random instances may not be representative of performance in practice, our algorithm is fast enough to support multi-gigabyte real-world use cases, such as~\cite{kris}.   

\begin{figure}
\begin{centering}
\includegraphics[width=4.5in]{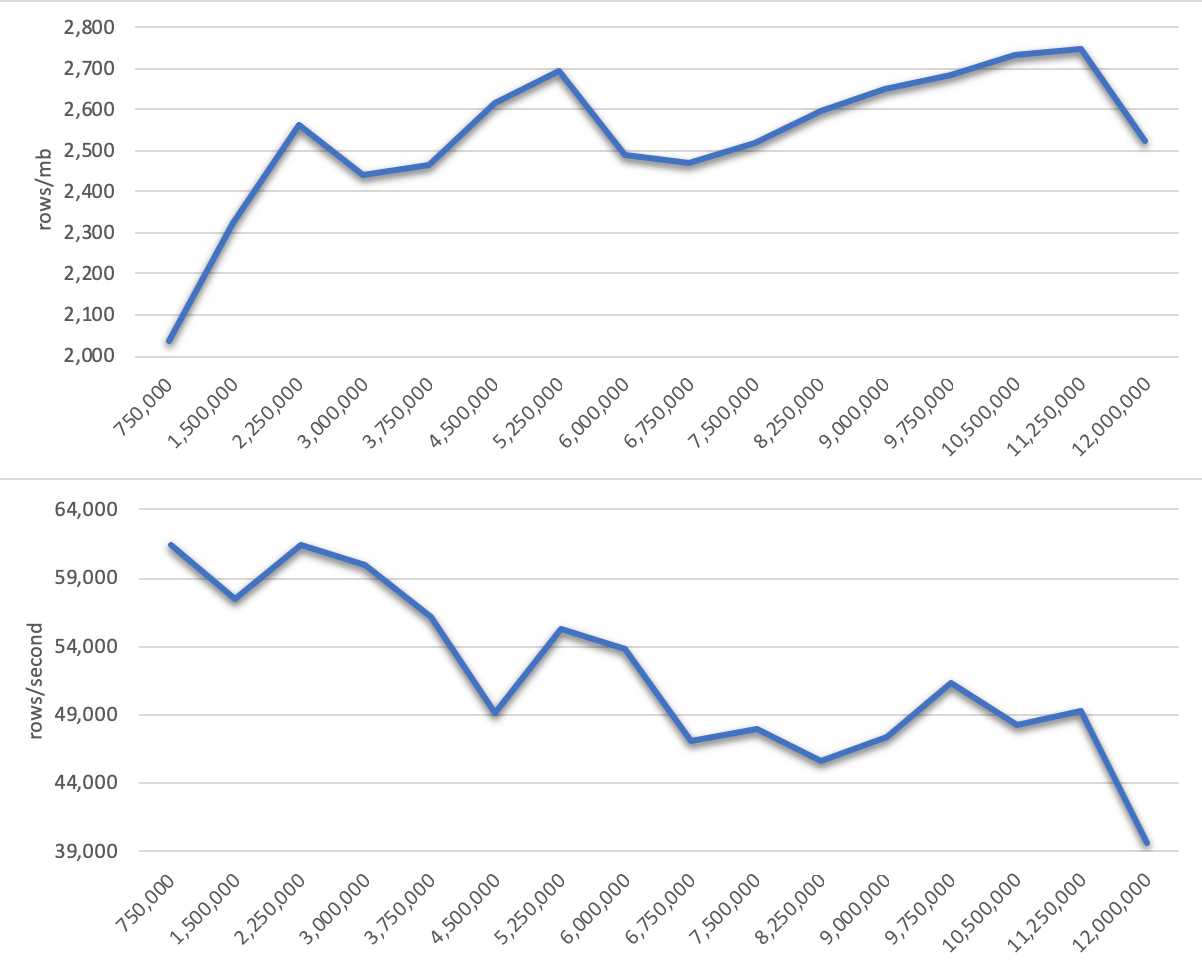}
\caption{Left Kan Chase Throughput, Quotient of a Set}
\label{fig4}
\end{centering}
\end{figure}

To demonstrate the significant speed-up of our algorithm compared to all the other algorithms we are aware of, Figure~\ref{fig5} shows time throughput for the same experiment using three {\it previous} left Kan algorithms: the ``substitute and saturate'' algorithm of~\cite{patrick} using either specialized Knuth-Bendix completion (``monoidal''~\cite{doi:10.1137/0214073}) or congruence closure~\cite{Nelson:1980:FDP:322186.322198} to decide the word problem associated to each category presentation, and the sequential chase-like left Kan algorithm of~\cite{BUSH2003107}.  All the algorithms are implemented in java 11 in the CQL tool, and share micro-level optimization techniques such as hash-consed terms~\cite{Baader:1998:TR:280474}, making the comparison relatively apples-to-apples, with one caveat: only our algorithm from this paper targets the skeleton of the category of sets, but with row counts limited to 12 million in this paper, the three algorithms besides our own that we compare to are all CPU-bound as opposed to memory bound, and so we hypothesize this difference does not impact our peformance analysis below.  

Performance analysis using java's built-in jvisualvm tool indicates that, as expected, the source of the performance benefit in our algorithm stems from the bulk-oriented (table at a time) nature of the actions that make up our rounds.  That is, the algorithms of ~\cite{CARMODY1995459} and ~\cite{BUSH2003107} are innately sequential in that they pick particular rows at a time non-determistically, and so their runtime comes to be dominated by many small sequentual reads and writes to large collections.  In contrast, our algorithm, by performing bulk-oriented collection operations, spends less time on row-level overhead.  This finding is consistent with that of the database theory literature, where parallel versions of the chase are deliberately employed because they are faster than sequential versions~\cite{onet:DFU:2013:4288}.

\begin{figure}
\begin{centering}
\includegraphics[width=4.75in]{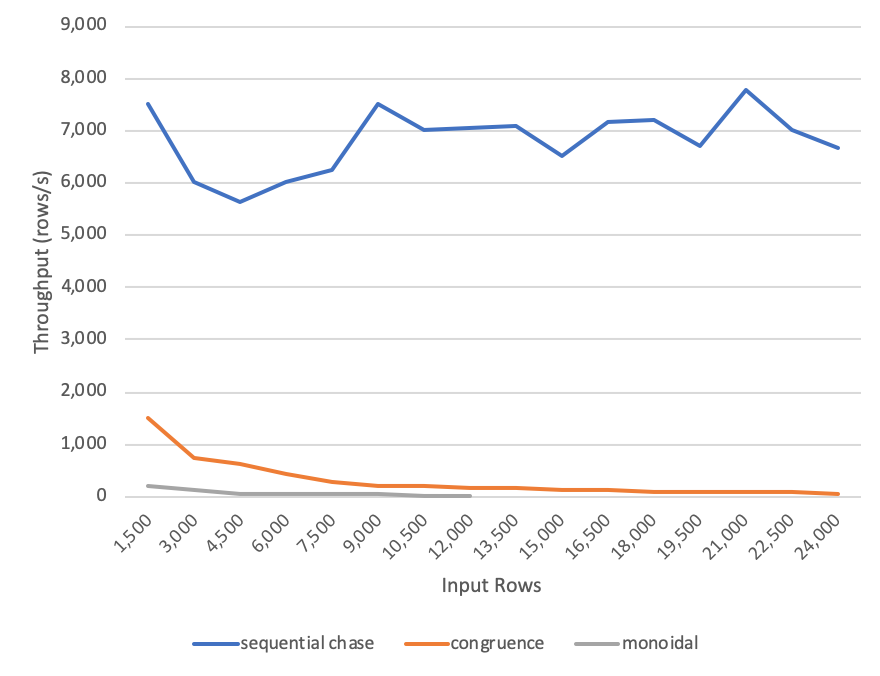}
\caption{Prior Left Kan Throughput, Quotient of a Set}
\label{fig5}
\end{centering}
\end{figure}


\section{Conclusion: Left Kan Extensions and Database Theory}
\label{comp}


We conclude by briefly summarizing how our use of the chase relates to its use in database theory.  Unlike traditional logic and model theory, where models hold a single kind of value (typically drawn from a {\it domain/universe of discourse}), in data migration, models / database instances hold two kinds of values: {\it constants} and {\it labelled nulls}.  Constants have inherent meaning, such as the numerals $1$ or $2$ or a Person's name; labelled nulls, sometimes called {\it Skolem variables}~\cite{Doan:2012:PDI:2401764}, are created when existential quantifiers are encountered during the chase and are distinct from constants and are not meaningful; they are considered up to isomorphism and correspond to the fresh {\sf g(v)} symbols in our left Kan algorithm.  All practical chase engines we know of enforce the constant/null distinction, and when an equality-generating dependency $n = c$ is encountered, where $n$ is a null and $c$ a constant, then $n$ is replaced by $c$, and never vice-versa; if $c = c'$ is encountered, where $c$ and $c'$ are distinct constants, then the chase {\it fails.}  In our general analysis of chase algorithms, we used the traditional model theoretic assumption that the input data was encoded entirely using labelled nulls, allowing us to sidestep the issue of chase failure.  We retain some of the semantic functionality of constants by our concept of ``weakly free model'', or alternatively by the device of uniquely satisfied unary relations (see Lemma~\ref{universalvsweaklyinitial}), yet these ``pseudo-constants'' can be merged, yielding an unfailing chase.  More a complication due to the fact that many categorical constructions cannot distinguish between isomorphic sets than a problem in practice, the consequences of adopting an unfailing, nulls-only chase procedure in the context of data migration are explored in~\cite{wadt,relfound,DBLP:journals/jfp/SchultzW17}.

However, in the particular case of the left Kan extensions, the input data is untouched, as it already satisfies all EDs whose ``backs'' have variables with input sorts.  The input data is copied to the output in step $2$ of Algorithm~\ref{alg:fastkan} and updated there.  Thus it would not affect the operation of Algorithm~\ref{alg:fastkan} if some or all of the input data were constants.  This observation extends to a general lemma.

\begin{lemma}\label{dataexchange}
Consider a {\it data exchange setting}, i.e. a signature $\sigma=(S_s\sqcup S_t,R_s\sqcup R_{st})$ and a theory $\ax=\ax_s\sqcup\ax_{st}$, where
\begin{itemize}
    \item $\sigma_s\coloneqq (S_s,R_s)$ is a signature; and,
    \item every ED in $\ax_s$ is a formula on $(S_s,R_s)$; and,
    \item every ED in $\ax_{st}$ has a conclusion all of whose conjuncts have a variable sorted in $S_t$.
\end{itemize}
Let $I$ be a $\sigma_s$-instance comprised of any combination of constants and labelled nulls (cf. the local definitions in part 3 of Lemma~\ref{universalvsweaklyinitial}), and consider it a $\sigma$-instance by letting $Io=\varnothing$ for sorts $o$ in $S_t$ and relation symbols $o$ in $R_{st}$.

If every input element is a constant, then the weakly free models of $\ax$ on $I$ are exactly the universal models of $\ax$ on $I$, so the standard and categorical core chases (see Section~\ref{core}) will yield the same result.
\end{lemma}

Thus when working in data exchange settings, ``weakly free model'' semantics is in no way less expressive than ``universal model'' semantics, yet it works better with category theory.  Also, restricting to cartesian theories (or making existing theories cartesian by replacing $\exists$ with $\exists!$) allows for completeness of the standard and parallel chase (given mild assumptions, see Proposition~\ref{chasecompleteness}) while making contact with the large body of category theory on initiality, adjunctions, reflective subcategories, finite-limit sketches, and essentially algebraic theories, as is foreshadowed by results on the so-called ``Skolem chase''~\cite{Benedikt:2017:BC:3034786.3034796}.

On the other hand, category theorists who want to consider databases would do well to pay more attention to the ``weak'' variants of notions, such as ``weak factorization systems''~\cite{nlab:weak_factorization_system,garner}, ``weakly reflective subcategories''~\cite{weakly_reflective}, and ``weak adjoints''~\cite{kainen_1971}.

A deeper integration between database theory and category theory can be reached if such conceptual shifts are made on both sides.

{\bf Intellectual Property.}  This paper is the subject of United States Letters Patent No. 11,256,672 granted February 22, 2022.

\bibliographystyle{spmpsci} 
\addcontentsline{toc}{section}{References}
 \bibliography{fastlk}
 
\section{Appendix}
 
\subsection{The Core Chase}
\label{core}

The purpose of this section is to show that the core chase could be used for computing left Kan extensions, since it too computes finite free models of cartesian theories, and it is complete (Lemma~\ref{corechase}).

The {\it core chase}~\cite{core} is a canonical (determined up to isomorphism) chase algorithm which is intractable (exponential time on each step~\cite{core}) but easy to work with in theory.

We first extend the notion of ``core'' in database theory.

\begin{definition}
Let $C$ be a category.  An object $c\in C$ is called {\it core} if whenever there is an object $c'\in C$, a morphism $c\to c'$, and a monomorphism $m:c'\hookrightarrow c$, we have that $i$ is an isomorphism.

A {\it core of} an object $c$ is a core object $c'$ with a monomorphism $m:c'\hookrightarrow c$ and a morphism $p:c\to c'$.

If $\sigma$ is a signature and $C=\sigma {\sf -Inst}$, then the definition above reduces to the usual definition of ``core''.

Now suppose that $I_0$ is an instance on $\sigma$ and $C=I_0/{\sf id}_{\sigma {\sf -Inst}}$.  We say that an instance $I$ is {\it core under $i:I_0\to I$} if the object $(I,i:I_0\to I)\in I_0/{\sf id}_{\sigma {\sf -Inst}}$ is core.  Explicitly, this means that whenever there is a subinstance $I'\subseteq I$ including the image of $i$ and a morphism $p:I\to I'$ such that $p\circ i = i$, we have that $I'=I$.

We also define in this case a {\it core of an instance $I$ under $i:I_0\to I$} as a core of $(I,i:I_0\to I)\in I_0/{\sf id}_{\sigma {\sf -Inst}}$. Explicitly, this is a subinstance $I'\subseteq I$ including the image of $i$ and a morphism $p:I\to I'$ such that $p\circ i = i$.

(Working with cores ``under $i$'' in this way is equivalent to considering the images of $i$ as constants.)
\end{definition}

Any universal statement about cores of instances over morphisms specializes to a statement about cores, since $\varnothing/{\sf id}_{\sigma {\sf -Inst}}\cong \sigma {\sf -Inst}$, where $\varnothing$ is the empty instance.

\begin{lemma}\label{bisurjective} If an instance $I$ is core under $i:I_0\to I$ and an instance $J$ is core under $j:I_0\to I$, there are morphisms $f:I\to J$ and $g:J\to I$ with $f\circ i = j$ and $g\circ j = i$, and $I$ is finite, then $f$ and $g$ are isomorphisms.
\end{lemma}
\begin{proof}
The morphism $g\circ f: I\to J\to g(J)$ satisfies $g\circ f\circ i = i$, so $g(J)=I$, and $g\circ f$ is surjective.  Since $I$ is finite, $g\circ f$ must be an isomorphism, so $f$ must be injective and $g$ must be surjective.  Similarly, $f\circ g:J\to I\to g(I)$ must be surjective, so $J$ is finite.  Thus $f\circ g$ is an isomorphism, so both $f$ and $g$ are.
\qed \end{proof}

\begin{lemma}
Every finite instance $I_0$ has a core under every morphism $i:I\to I_0$, and all of its cores under the same morphism $i$ are isomorphic.  Thus, we can speak of ``the core'' of $I_0$ under $i$, writing ${\sf core}_i(I_0)$.  
\end{lemma}
\begin{proof}
We first construct a core of $I_0$ under $i$.  If $I_0$ is already core, we are done.  Otherwise, there is a proper subinstance $I_1\subseteq I_0$ including the image of $i$ and a morphism $p_1:I_0\to I_1$ with $p_1\circ i = i$.  If $I_1$ is core, we are done.  Otherwise, there is a proper subinstance $I_2\subseteq I_1$ including the image of $i$ and $p_2:I_1\to I_2$ with $p_2\circ i=i$.  Since $I_0$ is finite, this sequence must terminate, resulting in a core subinstance $I_n$ under $i$ and a composite morphism $p$ of the path $I_0\To{p_1} I_1\To{p_2} I_2\To{p_3}\cdots \To{p_{n-1}} I_n$.  We have $p\circ i = i$, so $I_n$ is a core of $I_0$ under $i$.

Now consider two cores $p:I_0\to J$ and $q:I_0\to K$ of $I_0$ under $i$.  Let $f=q|_J:J\to K$ be the restriction of $q$ to $J$. and let $g=p|_K:K\to J$ be the restriction of $p$ to $K$.  By Lemma~\ref{bisurjective}, $f$ is an isomorphism.
\qed \end{proof}

The perceptive reader might have realized that the last lemma constructs merely an isomorphism of core instances, not an isomorphism of core instances along with their respective morphisms $p:I\to{\sf core}_i(I)$.  It turns out that the ordered pair $({\sf core}_i(I),p:I\to {\sf core}_i(I))$ is not defined up to isomorphism.

\begin{example}
Consider the uni-typed signature $\sigma$ with a single binary relation symbol $r$.  The instance $I$ with elements $\{{\sf foo},{\sf bar},{\sf baz}\}$ and $Ir=\{({\sf foo},{\sf bar})\}$ has core $J$ with elements $\{{\sf foo},{\sf bar}\}$ and $Jr=\{({\sf foo},{\sf bar})\}$.  We exhibit two distinct morphisms $I\to J$.  Let $p({\sf foo})={\sf foo},p({\sf bar})={\sf bar},p({\sf baz})={\sf foo}$, and $q({\sf foo})={\sf foo},q({\sf bar})={\sf bar},q({\sf baz})={\sf bar}$.  Then there are no isomorphisms $f$ and $g$ such that the following square commutes:
\[ \xymatrix{
I \ar[r]^f \ar[d]_p & I \ar[d]^q \\
J \ar[r]^g & J
} \]
Indeed, the only isomorphisms $f$ and $g$ are identities, so the commutativity of the square reduces to the false claim that $p=q$.
\end{example}

The assumption that $I$ is finite is also required for existence and uniqueness of cores.

\begin{example}
{\bf Existence fails in the infinite case:} Consider the uni-typed signature $\sigma$ with one binary relation symbol $r$ and the instance $N$ whose elements are the natural numbers $\mathbb{N}$ and where $Nr=\{(m,n)\mid m+1=n\}$.  This instance does not have a core.
\end{example}

\begin{example}
{\bf Uniqueness fails in the infinite case:} Consider the uni-typed signature $\sigma$ with one binary relation symbol $r$ and one unary relation symbol $a$.  Consider the instance $X$ whose elements are $\{n,n'\mid n\in\mathbb{Z}\}$ and where $Xr=\{(n,n+1),(n',(n+1)')\mid n\in\mathbb{Z}\}$ and $Xa=\{n,n'\mid n\geq 2\}\cup \{0'\}$.  Then the unprimed and primed components of $X$ are nonisomorphic cores of $X$.  See \cite[Theorem 29]{bauslaugh} for an example on the signature $\{r\}$.
\end{example}

\begin{definition}
The (standard) {\it core $\ax$-chase}~\cite{Deutsch:2008:CR:1376916.1376938} is a chase algorithm in which two steps alternate:
\begin{description}
\item[{\bf Parallel chase step}] $I_n\To{\mathcal{C}({\mathcal{F}_n})}I'_{n+1}$, where $\mathcal{F}_n$ is the set of all triggers in $I_n$ of EDs in $\ax$.
\item[{\bf Core step}] $I'_{n+1}\To{p_{n+1}} {\sf core}(I'_{n+1})\eqqcolon I_{n+1}$.
\end{description}

We also consider the {\it categorical core $\ax$-chase}, which is like the standard core chase except that it uses the following variant of the core step instead:
\begin{description}
\item[{\bf Categorical core step}] $I'_{n+1}\To{p_{n+1}} {\sf core}_{i'_{n+1}}(I'_{n+1})\eqqcolon I_{n+1}$, where $i'_{n+1}$ is the composite of the path $I_0\To{\mathcal{C}({\mathcal{F}_0})}I'_1\To{p_1}I_1\to\cdots\to I'_{n+1}$.
\end{description}

Intuitively, the categorical core chase differs from the standard in that it considers the images of the elements of the input instance as constants when doing core steps, but not when doing parallel chase steps.  The rationale for this is that we want to merge these input elements only when egds force us to, not just to make something core.

In either core chase, when the result of the core step is a model, the chase halts.  Note that both core chases are equivalent when $I_0=\varnothing$.  Just as in the standard and parallel chase, we use $\mathcal{C}(k,n)$ to denote the composite of the path $I_k\to\cdots\to I_n$.
\end{definition}

It is known~\cite[Theorem 7]{Deutsch:2008:CR:1376916.1376938} that given a regular theory $\ax$ and an input instance $I$, there is a finite universal model of $\ax$ on $I$ iff the core chase terminates and yields this model.  However, the core chase fails to compute finite weakly free models.  For example, take the uni-typed signature with no relation symbols, $\ax=\varnothing$, and $I_0=\{{\sf foo},{\sf bar}\}$.  Then the core chase terminates in one step: $I_0\To{p} I_1\coloneqq {\sf core}(I)=\{{\sf foo}\}$, and this is not weakly free, e.g. there is no morphism $f:I_1\to I_0$ such that $f\circ p = {\sf id}_{I_0}$.  The categorical chase computes weakly free models instead of universal models, but at the price of terminating on a strictly smaller class of inputs:  there are inputs (see Section~\ref{initweaklyinit}) which have finite universal models but all of whose weakly free models are infinite.

\begin{lemma}\label{corechase}
Let $\ax$ be a regular theory and consider a standard (categorical) core $\ax$-chase sequence $I_0\To{\mathcal{C}_0} I'_1\to I_1\to \cdots$.
\begin{enumerate}
    \item If this sequence terminates, then it computes a universal (weakly free) model of $\ax$ on $I_0$.  If we are using the categorical core chase and $\ax$ is cartesian, then we compute a free model of $\ax$ on $I_0$.
    \item If there is a finite universal (weakly free) model of $\ax$ on $I_0$, then this sequence terminates with a finite universal (weakly free) model of $\ax$ on $I_0$.
\end{enumerate}
\end{lemma}
\begin{proof}
$1$: Suppose this sequence terminates at $I_n$, and call the composition of the whole sequence $\mathcal{C}(0,n)$.  Then $I_n$ is a core model of $\ax$.  For a model $B$ of $\ax$ and a morphism $b_0:I_0\to B$, construct a morphism $b'_1:I'_1\to B$ with $b'_1\circ\mathcal{C}_0=b_0$ just as in Lemma~\ref{chasecomputesweaklyinit}.  Let $b_1:I_1\to B$ be the restriction of $b'_1$ to $I_1\subseteq I'_1$.  If we are doing the categorical core chase, then we have $p_1\circ \mathcal{C}_0=\mathcal{C}_0$, so $b_1\circ p_1\circ \mathcal{C}_0 = b_1\circ \mathcal{C}_0 = b'_1\circ \mathcal{C}_0 = b_0$.  Continuing this process by induction, we obtain $b_n:I_n\to B$, so $I_n$ is a finite universal model of $\ax$ on $I_0$.  If we are doing the categorical core chase, then we also obtain $b_n\circ \mathcal{C}(0,n) = b_0$ from the induction, so $(I_n, \mathcal{C}(0,n))$ is a finite weakly free model of $\ax$ on $I_0$.

If we are using the categorical chase and $\ax$ is cartesian, then let $(A,a)$ be a free model of $\ax$ on $I_0$ (using Corollary~\ref{existence}).  By ``weak free-ness'', there is a morphism $f:I_n\to A$ with $f\circ \mathcal{C}(0,n) = a$.  By ``freeness'' of $A$, there is a morphism $g:A\to I_n$ with $g\circ a = \mathcal{C}(0,n)$, and $f\circ g = {\sf id}_A$.  Thus $g(A)$ is a subinstance of $I_n$ with a morphism $g\circ f:I_n\to g(A)$.  Since $I_n$ is core, $I_n=g(A)$, so $g$ is an isomorphism and $(I_n,\mathcal{C}(0,n))$ is free.

$2$: Suppose there is a finite universal model $A$ (weakly free model $(A,a)$) of $\ax$ on $I_0$.  Without loss of generality we can choose $A$ to be core (under $a$).  Let $I$ be the colimit of the core chase sequence, with legs $l_n:I_n\to I,l'_n:I'_n\to I$.  By an argument similar to the proof of Lemma~\ref{coollemma}, $I$ is a universal model ($(I,l_0)$ is a weakly free model) of $\ax$ on $I_0$.  By universality (weak free-ness), there exist morphisms $f:A\to I$ and $g:I\to A$ (such that $f\circ a = l_0$ and $g\circ l_0 = a$).  By Lemma~\ref{filteredfinite}, there is an $n$ and a morphism $f_n:A\to I_n$ such that $l_n\circ f_n=f$ (and $f_n\circ a = \mathcal{C}(0,n)$).  Since $A$ is core (under $a$) and $I_n$ is core (under $\mathcal{C}(0,n)$) and we have morphisms $f_n:A\to I_n$ and $g\circ l_n: I_n\to A$ between them (with $f_n\circ a = \mathcal{C}(0,n)$ and $g\circ l_n\circ \mathcal{C}(0,n) = a$). Lemma~\ref{bisurjective} gives that $f_n$ is an isomorphism.
\qed \end{proof}


\subsection{Proof of Existence of Free Models Using the Ehresmann-Kennison Theorem}

We herein prove the ``cartesian'' half of Proposition~\ref{existence} a different way, using the Ehresmann-Kennison Theorem in the theory of sketches.

\begin{lemma}\label{existence2}
Given a signature $\sigma=(S,R)$ and  cartesian theory $\ax$, as in \cref{flt}, for any instance $I$ on  $\sigma$, there exists a free model $(\init_\ax(I), h)$ of $\ax$ on $I$. 

Moreover, $\init_\ax$ extends to a left adjoint to the forgetful functor $U:{\sf Mod}(\ax)\to \sigma{\sf -Inst}$, so ${\sf Mod}(\ax)$ is a reflective subcategory of $\sigma{\sf -Inst}$.
\end{lemma}
\begin{proof}
The proof uses the theory of sketches; see \cite{ttt} and \cite{Wells94sketches:outline}.  Note that understanding the proof is not required to understand our left Kan algorithm; in fact, we again prove this lemma using the chase in Lemma \ref{existence2}.  Given the signature $(S,R)$ and the theory $\ax$ there is a category $\mathfrak{C}$ and a set $\mathfrak{L}$ of cones such that the category of models for the sketch $(\mathfrak{C},\mathfrak{L})$ is equivalent to the category of $\ax$-models. Indeed, begin with the category with objects $S\sqcup R$ and a morphism $r \to s_i$ for each $r\in R$ with arity $s_0,\ldots, s_n$ and $0\leq i\leq n$.
Now for each $r\in R$ with arity $s_0,\ldots, s_n$, add an object $P_r$, a cone that forces it to be the product $s_0\times\cdots\times s_n$, and a cone that forces the canonical map $r\to P_r$ to be a monomorphism. Finally, for each axiom 
$$
  \forall (x_0:s_0) \cdots  (x_n:s_n) \ldotp
  \phi(x_0 , \cdots , x_n) \Rightarrow
  \exists ! (x_{n+1}:s_{n+1}) \cdots \ (x_m:s_m)\ldotp \psi(x_0, \ldots, x_m)
$$
in $\ax$, the conjunctions $\phi$ and $\psi$ are given by finite limits, say $p$ and $q$, which we introduce as new objects along with cones expressing that they are limits, and we finish by adding a morphism $p\to q$ and constraining it to commute correctly with projections.

By this process we obtain a finite-limit sketch $(\mathfrak{C}, \mathfrak{L})$, and it is tedious but not hard to show that the category of models of this sketch is equivalent to that of $\ax$-models.  


Now we use the Ehresmann-Kennison Theorem~\cite[Theorem 4.2.1]{ttt}, which states that the category of models of a finite limit sketch $(\mathfrak{C},\mathfrak{R})$ is a reflective subcategory of the functor category ${\sf Set}^{\mathfrak{C}}$, with reflector $R\dashv J$.

Now notice that the category $C$ constructed in Lemma~\ref{instancesascopresheafs} embeds in $\mathfrak{C}$ canonically --- call this embedding $i$.  Since $C$ is small, we then have an adjunction $\Sigma_i\dashv \Delta_i$.  We can combine the two adjunctions \cite{riehl} we have just constructed to form a composite adjunction $R\circ\Sigma_i\dashv\Delta_i\circ J$, as shown in the following diagram:

\[
\begin{tikzcd}
{\sf Mod}(\ax) \arrow[rr,"U", below]
&&
\sigma {\sf -Inst} \\
{\sf Mod}(\mathfrak{C},\mathfrak{L})
\arrow[r, "J"{name=J, below}, bend right=25]
\arrow[u, phantom, "\cong" rotate=-90]
&
{\sf Set}^{\mathfrak{C}}
\arrow[l, "R"{name=R,above}, bend right=25]
\arrow[phantom, from=J, to=R, "\dashv" rotate=-90]
\arrow[r, "\Delta_i"{name=Delta, below}, bend right=25]
&
{\sf Set}^C\arrow[l, "\Sigma_i"{name=Sigma, above}, bend right=25]
\arrow[phantom, from=Delta, to=Sigma, "\dashv" rotate=-90] 
\arrow[u, phantom, "\subseteq" rotate=-90]
\\
\end{tikzcd}
\]

For any instance $I$ on $\sigma$, considered as a functor $C\to{\sf Set}$, Lemma \ref{lemma:adjunction} gives that the category $I/(\Delta_i\circ J)$ has an initial object.  But the image of $\Delta_i\circ J$ is contained in the full subcategory $\sigma{\sf -Inst}$ (see Lemma~\ref{instancesascopresheafs}), so it lifts to $U$.  Thus we have an equivalence $I/(\Delta_i\circ J)\cong I/U$, proving the lemma.

\qed\end{proof}

\subsection{Semi-Na\"{i}ve Optimization}

In this section we discuss the ``semi-na\"{i}ve'' optimization for Algorithm~\ref{alg:fast}.  In this optimization, we avoid considering triggers that have, so to speak, already been considered.  The description of this optimization in \cite{Benedikt:2017:BC:3034786.3034796} relies on a particular construction of chase steps, but in our treatment thus far we have abstracted away from the construction of chase steps, opting instead to define them up to isomorphism, through a universal property.  We will stay abstract and bring into our abstraction only what is needed.

\begin{definition}
Let $\sigma$ be a signature. An {\it edit} of $\sigma$-instances $I$ and $J$ is a commutative diagram
\[
\begin{tikzcd}
I \arrow[rr, "f"] && J \\
&A \arrow[ul,hookrightarrow, "i"] \arrow[ur, hookrightarrow, "j"{below}]
\end{tikzcd}
\]
where $i$ and $j$ are monic, and we have $f\circ j = i$.
We write this edit succinctly as $(f,A,i,j):I\Longrightarrow J$.
\end{definition}

The instance $A$ expresses the data which is unchanged in the edit, occurring in both $I$ (via the embedding $i$) and $J$ (via the embedding $j$).  The edit is executed by discarding the additional data in $I$ and adding in the additional data in $J$.

To discuss semi-naive optimization, we must envision chase steps as edits rather than mere morphisms.  A morphism can be trivially upgraded to an edit by setting $A=\varnothing$, but standard implementations will do better than this.  For example, a chase step for a trigger of a tgd can easily be implemented with $A=I$ and $i=\textrm{id}_I$.  Loosely, the larger that $A$ can be made, the more the algorithm can be optimized.

\pagebreak 

\begin{definition}
Let $\sigma\textrm{-Edit}$ be the category of $\sigma$-instances and edits.  In this category, edits are composed through a pullback:

\[
\begin{tikzcd}
I \arrow[rr, "f"] && J  \arrow[rr, "g"] && K \\
&  A \arrow[ul,hookrightarrow, "i"] \arrow[ur, hookrightarrow, "j"{below}]
&& B \arrow[ul,hookrightarrow, "j'"] \arrow[ur, hookrightarrow, "k"{below}] \\
&& A\cap B \arrow[ul,hookrightarrow] \arrow[ur, hookrightarrow] \arrow[uu, phantom, "\rotatebox{135}{\scalebox{1.5}{$\lrcorner$}}" , very near start, color=black]
\end{tikzcd}
\]

And for any instance $I$, we have an ``identity'' edit $\overline{\textrm{id}}_I$, defined as the tuple $(\textrm{id}_I,I,\textrm{id}_I,\textrm{id}_I)$.
\end{definition}

\begin{definition}
Consider an edit $e\coloneqq (u,A,i,j):I\Longrightarrow J$ and a morphism $f:K\to J$.  We say that $f$ is $e$-{\it old} if the image of $f$ is contained in the image of $j$.  Otherwise, we say that $f$ is $e$-{\it new}.

\end{definition}

Equivalently, $f$ is $e$-old if it factors through $i$, as shown:
\[
\begin{tikzcd}
I \arrow[rr, "f"] && J \\
&A \arrow[ul,hookrightarrow, "i"] \arrow[ur, hookrightarrow, "j"{below}] \\
&& K \arrow[uu,"f"{right}] \arrow[dashed, ul, "g"] \\
\end{tikzcd}
\]
In this diagram we require that $g\circ j = f$.

For the semi-na\"{i}ve optimization, we introduce the variables $e_{\textrm{tgd}}$ and $e_{\textrm{egd}}$ to keep track of tgd and egd chase steps so as not to re-consider triggers that have already been chased.  In steps 5 and 12, we can see that triggers are required to be new with respect to these edits before it is even checked whether or not they are active.  We also split the variable $I_{\textrm{cur}}$ occurring in Algorithm~\ref{alg:fast} into $I_{\textrm{tgd}}$ and $I_{\textrm{egd}}$ for the more sophisticated bookkeeping necessary here.  The variables $I_{\textrm{tgd}}^{\textrm{prev}}$ and $I_{\textrm{egd}}^{\textrm{prev}}$ are entirely for the purpose of the discussion in Proposition~\ref{seminaive} --- they are not necessary for the algorithm to work.

\begin{algorithm}[H]\label{alg:seminaive}
\setstretch{1.35}
\SetAlgoLined
\caption{Semi-Na\"{i}ve Fast Parallel Chase Algorithm}\label{alg:seminaive}
\KwData{$\sigma$, $\ax$, $I_0$}
\KwResult{$I$, $i:I_0\to I$}
$i\coloneqq$ parallel chase $I_0\xRightarrow{\mathcal{C}(\textrm{all active triggers of EDs in $\ax$ with empty }{\sf front})} I$\;
$e_{\textrm{tgd}} \coloneqq (\varnothing \xRightarrow{(!,\varnothing,!,!)} I)$\;
$I_{\textrm{tgd}}\coloneqq I$\;
${\sf first}\coloneqq {\sf true}$\;
\While{$\textrm{TGD}\coloneqq \{\textrm{triggers } f \textrm{ of tgds in }I_{\textrm{tgd}} \mid f\textrm{ is $e_{\textrm{tgd}}$-new} \wedge f\textrm{ is active}\}$ is nonempty OR ${\sf first}$}{
  parallel chase $I_{\textrm{tgd}}\xRightarrow{\mathcal{C}(\textrm{TGD})}I$\;
  $i\coloneqq\textrm{compose}(I_0\xRightarrow{i}I_{\textrm{tgd}}\xRightarrow{\mathcal{C}(\textrm{TGD})} I)$\;
  $e_{\textrm{tgd}}\coloneqq \mathcal{C}(\textrm{TGD})$\;
  $e_{\textrm{egd}}\coloneqq {\sf first} ? (!,\varnothing,!,!) : \mathcal{C}(\textrm{TGD})$\;
  $I_{\textrm{egd}}^{\textrm{prev}}\coloneqq I_{\textrm{egd}}$\;
  $I_{\textrm{egd}}\coloneqq I$\;
  \While{$\textrm{EGD}\coloneqq \{\textrm{triggers } f \textrm{ of egds in }I_{\textrm{egd}} \mid f\textrm{ is $e_{\textrm{egd}}$-new} \wedge f\textrm{ is active}\}$ is nonempty}{
    parallel chase $I_{\textrm{egd}}\xRightarrow{\mathcal{C}(\textrm{EGD})}I$\;
    $i\coloneqq\textrm{compose}(I_0\xRightarrow{i}I_{\textrm{egd}}\xRightarrow{\mathcal{C}(\textrm{EGD})} I)$\;
    $e_{\textrm{tgd}}\coloneqq\textrm{compose}(I_{\textrm{tgd}}\xRightarrow{e_{\textrm{tgd}}}I_{\textrm{egd}}\xRightarrow{\mathcal{C}(\textrm{EGD})} I)$\;
    $e_{\textrm{egd}}\coloneqq \mathcal{C}(\textrm{EGD})$\;
    $I_{\textrm{egd}}^{\textrm{prev}}\coloneqq I_{\textrm{egd}}$\;
    $I_{\textrm{egd}}\coloneqq I$\;
  }
  $I_{\textrm{tgd}}^{\textrm{prev}}\coloneqq I_{\textrm{tgd}}$\;
  $I_{\textrm{tgd}}\coloneqq I$\;
  ${\sf first}\coloneqq {\sf false}$\;
}
\end{algorithm}

\begin{proposition} \label{seminaive}
Assume that the parallel chase step is implemented as an edit, not merely a morphism.  Then Algorithm~\ref{alg:seminaive} computes the same function as Algorithm~\ref{alg:fast}, and both functions converge for exactly the same inputs.
\end{proposition}
\begin{proof}
It is apparent that if the conjunctive conditions in lines 5 and 12 of Algorithm~\ref{alg:seminaive} were replaced with simply ``$f$ is active'', the algorithm would become essentially the same as Algorithm~\ref{alg:fast}.  Thus it suffices to show that these conditions are equivalent to ``$f$ is active'' whenever these lines are run, i.e. that the first conjunct fails only on inactive triggers.

First we take line 5.  If $f:{\sf front}\to I_{\mathrm{tgd}}$ is a trigger of a tgd in $I_{\mathrm{tgd}}$ which is $e_{\mathrm{tgd}}$-old, we must show that it is inactive. The first time that this line is encountered, $e_{\mathrm{tgd}}$ is equal to $(!,\varnothing,!,!)$ as given in line 2, so we must have ${\sf front}=\varnothing$.  We see from the following diagram that $f$ is inactive, since the dashed arrow exists by line 1.

\begin{tikzcd}
\varnothing \ar[rd, "f"] \ar[d] \ar[r] & I_0 \ar[d, "i"] \\
{\sf back} \ar[r, dashed] & I_{\textrm{tgd}}
\end{tikzcd}

The following times that line 5 is encountered, $e_{\mathrm{tgd}}$ is given by lines 9 and 15 --- it is the composite of the chase from $I_{\textrm{tgd}}^{\textrm{prev}}$ to $I_{\textrm{tgd}}$.  Since $f$ is $e_{\mathrm{tgd}}$-old, it factors through $I_{\textrm{tgd}}^{\textrm{prev}}$; thus by line 6 we have the dashed arrow in the following diagram, so $f$ is inactive.

\begin{tikzcd}
&& I_{\textrm{tgd}}^{\textrm{prev}} \ar[d, "\mathcal{C}(\textrm{TGD})"] \\
{\sf front} \ar[rrd, "f" description] \ar[r] \ar[d] & A \ar[ur, hookrightarrow] \ar[dr, hookrightarrow] & \cdot \ar[d] \\
{\sf back} \ar[rru, dashed, bend right=15] && I_{\textrm{tgd}}
\end{tikzcd}

Now we take line 12.  If $f:{\sf front}\to I_{\mathrm{egd}}$ is a trigger of an egd in $I_{\mathrm{egd}}$ which is $e_{\mathrm{egd}}$-old, we must show that it is inactive.  The first time this line is encountered, $e_{\mathrm{tgd}}$ is equal to $(!,\varnothing,!,!)$ as given in line 9, so we must have ${\sf front}=\varnothing$.  Since the morphism ${\sf front}\to{\sf back}$ must be surjective for a tgd, we have ${\sf back}=\varnothing$ as well, so $f$ is clearly inactive.  If line 12 is encountered an additional time at the beginning of the ``egd'' while loop, we have $e_{\mathrm{egd}}=\mathcal{C}(\textrm{TGD})$.  Then $f$ factors through $I_{\textrm{tgd}}$, as shown in the following diagram.  By the inner while loop (lines 12-18), $I_{\textrm{tgd}}$ satisfies all egds, giving us the dashed arrow.  Thus $f$ is inactive.

\begin{tikzcd}
&& I_{\textrm{tgd}} \ar[dd, "\mathcal{C}(\textrm{TGD})"] \\
{\sf front} \ar[rrd, "f" description] \ar[r] \ar[d] & A \ar[ur, hookrightarrow] \ar[dr, hookrightarrow] & \\
{\sf back} \ar[rruu, dashed, bend right=40] && I_{\textrm{egd}}
\end{tikzcd}

Finally, we have the case that line 12 is encountered an additional time within the inner while loop.  In this case, $e_{\mathrm{egd}}=\mathcal{C}(\textrm{EGD})$; then $f$ factors through the previous value of $I_{\textrm{egd}}$.  The dashed arrow in the following diagram then exists by line 13, so $f$ is inactive.

\begin{tikzcd}
&& I_{\textrm{egd}}^{\textrm{prev}} \ar[dd, "\mathcal{C}(\textrm{EGD})"] \\
{\sf front} \ar[rrd, "f" description] \ar[r] \ar[d] & A \ar[ur, hookrightarrow] \ar[dr, hookrightarrow] & \\
{\sf back} \ar[rr, dashed] && I_{\textrm{egd}}
\end{tikzcd}

\qed\end{proof}

\tableofcontents

 \end{document}